\documentclass{article}
\usepackage[T1]{fontenc}
\usepackage{graphicx}
\usepackage{color}
\usepackage{hyperref}

\usepackage{algorithm}
\usepackage{algorithmic}
\usepackage{amsfonts}
\usepackage{amsmath}
\usepackage{amsthm}
\usepackage{amssymb}
\usepackage{booktabs}
\usepackage{colortbl}
\usepackage{float}
\usepackage{multirow}
\usepackage{subcaption}
\usepackage{url}           
\usepackage{enumitem}
\newtheorem{theorem}{Theorem}
\newtheorem{lemma}{Lemma}
\newtheorem{remark}{Remark}
\newtheorem{assumption}{Assumption}
\newtheorem{definition}{Definition}

\usepackage[english]{babel}

\usepackage[textsize=tiny]{todonotes}

\usepackage{amsmath}
\usepackage{graphicx}

\newsavebox{\btcMLEtablebox}

\usepackage{tikz}
\usetikzlibrary{arrows,shapes,automata,petri,positioning,calc}
\tikzset{
    place/.style={
        circle,
        thick,
        draw=black,
        fill=gray!50,
        minimum size=6mm,
    },
        state/.style={
        circle,
        thick,
        draw=blue!75,
        fill=blue!20,
        minimum size=6mm,
    },
}

\title{Performance Guarantees\\ for Permissionless Lending}

\author{Nanqiao Wang\thanks{email: wangnanqiao2022@gmail.com}~ and Jia Yuan Yu\thanks{email: jy@osore.ca}}

\date{}
\begin{document}
\maketitle

\begin{abstract}
We propose a lending protocol with far fewer risks than existing pool-based lending protocols (e.g., Aave, Compound). In our protocol, lenders and borrowers embody the spirit of decentralization by making competing interest-rate offers, resulting in the discovery of the prevailing interest rate through  t\^atonnement-like dynamics as opposed to taking rates partly determined by a central authority (e.g., Aave, Federal Reserve).
We eliminate the risk of manipulation by creating loans with a fixed interest rate throughout the loan---like US Treasury bills.
In particular, we further derive performance guarantees probabilities of repayment shortfall and liquidation by modeling the value of the collateral as a random walk, thus allowing participants to more efficiently manage their risk. To do so, we use characteristic functions and Fourier inversion to derive exact integral representations of the liquidation and repayment shortfall probabilities under the most general assumptions.
We illustrate empirically 
this method by using Bitcoin as collateral and fitting its historical prices to Gaussian, Student-$t$, and symmetric stable Paretian distributions.
\end{abstract}




\section{Introduction}

Risk-averse decision-makers naturally allocate capital across multiple assets, some of which are less liquid (e.g., wrapped tokens, staked tokens, BTC, tokenized RWAs) than others (e.g., stablecoins, US treasury bills).
When these decision-makers come across some information and need quick access to liquid assets, they often resort to using the less liquid asset as collateral to borrow liquid assets. However, existing permissionless lending protocols (e.g., Aave, Compound) have two major weaknesses. First, they are tremendously opaque and lack guarantees on risks of liquidation and repayment shortfall. Second, they present extraordinary risks such as arbitrarily-set interest rates, liquidation attacks by manipulation of utilization rates, front-running, withdrawal bottleneck \cite{rseth}, and overly conservative metrics for valuating the collateral.

We propose a permissionless asynchronous marketplace for loans that addresses these two weaknesses. To address the second weakness, this marketplace operates in a simpler fashion---like Facebook Marketplace and Prosper---where anyone can offer to lend out an amount of liquid assets for a fixed duration and at a fixed interest rate.
Moreover, anyone can take up an existing offer by putting down an appropriate collateral, and vice versa.
Participants continuously compete by looking at each others’ offers and updating their own offers.
Loans are concluded in an asynchronous fashion---on a first-come-first-served basis like order books---without the strict timing of auctions.
Since all parties are trustless, loans are over-collateralized. Every borrower must put down a collateral that exceeds in value what is borrowed. 
Just as stablecoins are characterized by their predictable exchange rates, our loans are characterized by fixed interest rates impervious to third-party manipulation.

To address the first weakness, we provide very general performance guarantees for participants in our protocol.
It goes without saying that, just as serious blockchains require safety and liveness guarantees \cite{gkl}, a useful lending protocol that handles billions of dollars of value must come with performance guarantees. 
Although we would like to guarantee that participating in our protocol is more efficient than in other decentralized finance protocols---e.g., staking \cite{beccuti2025formalframeworkethereumstaking}, minting algorithmic stablecoins  \cite{strohmeyer2024improvingdefimechanismsdynamic}, or swapping on decentralized exchanges \cite{he2025arbitragedecentralizedexchanges},  such a strong guarantee would inevitably require strong assumptions.
Instead, we settle for guaranteeing that, with high probability, a loan will lead to an outcome (repayment with interest) that is agreeable to both the lender and the borrower. This is possible as long as the  value of the collateral follows some particular trajectories.

The importance of guarantees is often overlooked in the design of decentralized finance protocols. However, just like the safety and liveness guarantees for ledger-type and virtual-machine-type blockchains are central to their success, similar guarantees are absolutely essential if a permissionless protocol is to manage billions of dollars of value.
Despite the extent of use of existing lending pools, some features are still quite counterintuitive.
For instance, to avoid liquidation, increasing collateral may be more effective than repaying a loan \cite{bl25}.
Moreover, defaulting on a loan is sometimes the best choice when the health factor is too low.
These surprising insights highlight the importance of giving performance guarantees to users of lending protocols.

Lending protocols serve two asymmetric sets of users: lenders requiring a guarantee of repayment with interest, and borrowers requiring the return of their collateral.
The lender requires a guarantee that most of the time, it will walk away from the loan with the principal and interest, with the exception of rare events, whose probability is controlled by the volatility of the collateral valuation.
The borrower requires a guarantee that most of the time, the loan concludes with the return of the collateral minus repayment and interest, and with the exception of rare events.
Our protocol makes available to the lender and borrower a set of actions (liquidation and top-up) that control the probability of these rare events.

Our analysis begins by formally specifying the participants, actions, states, and state-transition rules of the proposed lending protocol in Section~\ref{sec:marketplace}. Under i.i.d. Gaussian returns, we derive exact expressions for borrower liquidation and lender repayment shortfall probabilities in Section~\ref{sec:guarantee}. In Section~\ref{sec:non_gaussian_liquidation},
we first express the liquidation time as a first-passage time of a nonnegative running maximum. We follow the approach of \cite{ohgren2001remark,petrella2004numerical} and  use Spitzer's Theorem \cite{spitzer1956combinatorial}
and Pinelis' Integral Formula \cite[Eq.~16]{pinelis2015characteristic} to derive a sequence of convolutions that computes 
this running maximum's characteristic function from the return distribution.
Finally, we apply Gil-Pelaez Fourier Inversion \cite{gilpelaez1951inversion}
on this characteristic function to derive the liquidation probability. For the lender repayment probability, we combine the joint characteristic function recursion of \cite{petrella2004numerical} with Fourier inversion to derive an exact integral representation.
%
Since first passage times are encountered everywhere from queuing theory to probability theory, Theorem~\ref{thm:iid_increment_liquidation_probability} can be useful outside of the finance domain.
We situate our work within the literature in Section~\ref{sec:related} and conclude by listing open problems for future work.



\section{Proposed Lending Protocol}\label{sec:marketplace}

In this section, we formally present a lending protocol inspired from existing protocols, but with enough differences that allow us to give performance guarantees on individual loans. 
First, we define a pair of tokens: a liquid token (e.g., a stablecoin such as USDC) and a less liquid token (e.g., BTC, tokenized assets). Next, we each participant has a type, as in Bayesian games: the type determines whether a participant is a lender holding the liquid token or a borrower\footnote{The borrower is a ``leverager'' who uses a collateral as leverage to secure a loan of the liquid asset.} holding the less liquid token, and determines the private information (or belief) held by each participant. The participants interact through sequential actions, and these actions are observed by other participants (e.g., committed to the blockchain).

Just as stablecoins are characterized by their predictable exchange rates, fixed-rate loans are characterized by predictable interest rates.
Our lending protocol is a peer-to-peer marketplace for loans, like Facebook Marketplace and Prosper, where anyone can offer to sell items with arbitrary sale parameters (e.g., cash payment, pick-up in person, etc.). 
Moreover, buyer can also offer to purchase items that meet a set of criteria.
Sellers and buyers compete by looking at each others’ offers and updating their own offer parameters.
A sale is concluded instantaneously on a first-come-first-served basis.
The main difference with Facebook marketplace and traditional lending through banks is the need for over-collateralization in the trustless world of DeFi. Every borrower must put down a collateral that exceeds in value what is borrowed. The concept of value is tied to a pricing function, and it is hard to get right, especially for tokens with lower circulation and liquidity.

Despite the large number of participants, we will use notation borrowed from Bayesian games and stochastic games \cite{shapley} to facilitate the description of the interaction between a number of participants (wallets) subject to a set of rules and state transitions. The state tracks who owns what, the rules determine the actions allowed in each state and the resulting next state.

\subsection{Time, Assets, Wallets, Valuation}

Throughout the paper, we consider two assets: a liquid asset \textsc{liq} (e.g., USDC) and a collateral asset \textsc{col} (e.g., BTC, tokenized assets) that is illiquid or less liquid.
The liquid asset also serves as the num\'eraire. 
Although the notion of liquidity of an asset does not directly affect the guarantees that we present, it is reflected in the price of the asset \cite{drechsler2020liquidity,kiyotaki2005liquidity,farhi2008theory}: an asset is less liquid when its price is more volatile (e.g., due to private information).
Time is discrete and indexed by $t=0,1,2,\ldots$, corresponding to successive checkpoints at which the lending protocol evaluates the states of ongoing loans. Each participant is a cryptographic wallet characterized by a type that determines the assets held, whether it is a lender or borrower, and its preference over the space of outcomes (e.g., long-term discounted payoff). In particular, the type determines whether the participant prioritizes one loan parameter value (e.g., interest rate $r$) over another (e.g., liquidation threshold $c$, etc.). 
Since decentralized lending protocol are permissionless, a single decision-maker can easily control multiple wallets with the same or different types---e.g., both lender and borrower wallets.

For notation sake, we distinguish three wallets:
lender wallet $w_L$, borrower wallet $w_B$, and vault wallet $w_V$.
The lender wallet $w_L$ and borrower wallet $w_B$ are participant-controlled wallets. 
The vault wallet $w_V$ denotes a wallet created by the protocol for the sole purpose of storing the liquid and collateral assets committed to an offer or loan. It is not a participant in the game-theoretic sense: it has no type and no preference relation. Its balance of assets changes only through the state-transition rules of the protocol.
The creation of the wallet $w_V$ is triggered by certain transactions of the protocol, i.e., offer-to-lend and offer-to-borrow of Section~\ref{sec:act}. Each vault wallet $w_V \triangleq w_V(\gamma)$ is uniquely associated with a transaction hash $\gamma$. We drop the $\gamma$ notation when the triggering transaction is clear from context.

The balance of each wallet at time $t$ is represented by a pair of nonnegative numbers. For example, we denote the lender wallet $w_L$'s balance at time $t$ by
    $W_{L,t} = (W_{L,t}[0], W_{L,t}[1]) \in \mathbb R_{\ge 0}^2$.
The coordinate $[0]$ denotes the collateral balance, and the coordinate $[1]$ denotes the liquid asset balance. The same convention applies to $w_B$ and $w_V$. Thus, $W_{B,t}[0]$ is the collateral balance of the borrower wallet at time $t$, and $W_{V,t}[1]$ is the liquid balance of the vault wallet at time $t$. Wallet balances are generally random variables, since they depend on participant actions and on the future valuation of collateral.

This wallet-based representation is a major departure from lending pools (cf. 
Appendix~\ref{sec:pool}). In the lending-pool model, credit-tokens and debit-tokens are introduced to track deposits and outstanding debt. In our protocol, there are no credit-tokens or debit-tokens. Each offer or loan is instead represented directly by its record and by the balances of its associated vault wallet.

We now define the valuation process. Since liquid asset is the num\'eraire, the relevant valuation is the value of the collateral in the liquid asset. At time $t$, let $\Xi_t \in \mathcal X$ denote a snapshot of the external environment, such as information from centralized exchanges and decentralized exchanges. The valuation process takes $\Xi_t$ and an amount $z$ of collateral as inputs, and returns the corresponding value in liquid asset: $\phi_t(\Xi_t,z) \in \mathbb R_{\geq 0}$.
When the environment $\Xi_t$ is clear from context, we write simply $\phi_t(z)$. Unlike a unit-price oracle, this valuation process can depend on the quantity $z$ being valued. This allows the model to capture liquidity and market-depth effects when the collateral amount is large. We assume that, for each $t$ and $\Xi_t$, the map $z \mapsto \phi_t(\Xi_t,z)$ is monotone increasing. Hence, when the inverse is well-defined, $\phi_t^{-1}(x)$ denotes the amount of collateral whose value at time $t$ is $x$ liquid asset.

\begin{remark}[Valuation convention for the liquid asset]
\label{rem:liq_valuation_convention}
Throughout this paper, we take \textsc{liq} to be USD and express all
asset values in USD. We assume that one USD is worth one USD.
Consequently, the USD value of an amount $v$ of \textsc{liq} is also
$v$. We therefore use $v$ to denote both the amount of \textsc{liq} and
its USD value, without introducing a separate valuation mapping for
\textsc{liq}. By contrast, $\phi_t(z)$ denotes the USD value at time $t$
of an amount $z$ of \textsc{col}.
\end{remark}

\subsection{Design Principles}

We design the set of actions to meet the following six fundamental principles.
\begin{enumerate}
\item[\textsc{pcp} 1] During an interaction where the lender acts first, and the borrower acts second:
\begin{enumerate}
    \item The lender controls the loan amount $v$, the loan duration $d$, the interest rate $r$, and the probability of repayment with interest.
    \item The borrower controls the start time $t_0$ and the probability of liquidation (through the collateral amount $z$).
\end{enumerate}

\item[\textsc{pcp} 2] During an interaction where the borrower acts first, and the lender acts second: 
\begin{enumerate}
\item The borrower controls the loan amount $v$, the collateral amount $z$, the loan duration $d$, the interest rate $r$, and influences the probability of liquidation.
    \item The lender controls the start time $t_0$.
\end{enumerate}

    \item[\textsc{pcp} 3] (Ideal ending) The borrower can end the loan at any time by repaying in full with interest.
    \item[\textsc{pcp} 4] (Painless ending) The lender or the borrower can split the collateral at any time after maturity, with a share equal in value to the repayment amount going to the lender.

    \item[\textsc{pcp} 5] During the course of the loan, the borrower controls the probability of liquidation by topping-up the collateral.
    \item[\textsc{pcp} 6] (Painful ending) During the course of the loan, the lender controls the probability of repayment by liquidating the collateral when some conditions are met.
\end{enumerate}

\begin{remark}[Similarity to US Treasuries]
    Principle 2 is exactly how US treasuries are sold to retail investors at non-competitive interest rates.
    Principle 1 is similar to Principle 2, but giving more control to the borrower.
\end{remark}

\subsection{Actions, Offers, Loans}\label{sec:act}

In our protocol, it is more natural to introduce the action space before the state space.
The actions available are reminiscent of those in a limit order book \cite{cont2023mathematicalframeworkmodellingorder}, with a few additional actions, and more parameters for each action because each loan is not only described by a price.
In this section, we give only an overview of actions, and defer to Section~\ref{sec:transition} the description of conditions under which each action is allowed.
Each participant has a private type and observes the set of available offers and wallet balances. Next, each participant decides to take-up an existing offer, wait for a better offer, or create its own offer. A participant who creates an offer is a first-mover who must be careful to include a risk premium in their offer that accounts for the risk that other participants possess private information, e.g., about an impending crash in value of a collateral \cite{rseth}. This risk premium can be in the form of a higher interest rate, a lower liquidation threshold higher, etc.
A participant who is already a party in an ongoing loan must also decide whether to top-up loans where it is a borrower, or to enter into a new loan with better terms in order to repay a previous loan, and whether to liquidate loans where it is lender.

Beside waiting and taking no action, at every time step, a participant can take an arbitrary number of actions selected from the following set, which fulfill the fundamental requirements\footnote{If we ignore the transaction limits in one block of a blockchain, the action space is actually countably infinite.}. The parenthetical \textsc{pcp} references indicate the design principles associated with each action.
\begin{enumerate}
    \item\label{act:offer-to-lend} (\textsc{pcp} 1) Offer to lend: a lender wallet \(w_L\) makes public an offer to lend out \(v\) units of \textsc{liq}, with a fixed duration \(d\), a fixed interest rate \(r\) for each time step, an initial \textsc{col} minimum required collateral ratio \(\rho\), and liquidation condition \(c\). The protocol creates a vault wallet \(w_V\) for this offer, and the vault wallet holds the \(v\) units of \textsc{liq} until the offer is taken up or canceled.
    The liquidation condition \(c\) specifies that whenever the ratio of values between collateral \textsc{col} and loaned \textsc{liq} drops below the liquidation threshold \(c\), liquidation may be triggered.
    For simplicity, we omit for now additional contract conditions imposed by the lender, such as limiting the set of further actions that can be taken on the borrowed tokens.
    \item\label{act:offer-to-borrow} (\textsc{pcp} 2) Offer to borrow: a borrower wallet $w_B$ makes public an offer to borrow $v$ units of \textsc{liq}, backed by $z$ units of \textsc{col},
with a fixed duration $d$, a fixed interest rate $r$ per time step,
a minimum required initial collateral ratio $\rho$, and a liquidation
threshold $c$. The protocol creates a vault wallet $w_V$ for this offer,
and the vault wallet holds the $z$ units of \textsc{col} until the offer
is taken up or canceled.
    \item\label{act:take-up-offer} (\textsc{pcp} 1 and 2) Take-up an offer: a wallet takes up an offer for lending or an offer for borrowing, with collisions resolved according to first-come-first-served or transaction precedence as recorded on the blockchain.
    \item\label{act:cancel-offer} (\textsc{pcp} 1 and 2) Cancel previous offer: a wallet cancels a previously announced offer that has not yet found a taker.
    \item\label{act:top-up} (\textsc{pcp} 5) The borrower can top-up the amount of collateral in an ongoing loan to reduce the perceived probability of defaulting.
    \item\label{act:interrupt-loan} (\textsc{pcp} 3) Interrupt loan: the borrower can end the loan early by repaying the lender the principal plus interest over the full duration \(d\).
    \item\label{act:natural-end} (\textsc{pcp} 4) If the loan is still ongoing at the maturity time, the liquidation condition is first evaluated. If the
    liquidation condition is not satisfied, the lender or, with the lender's
    authorization, the protocol service provider immediately executes the
    natural-ending settlement using the collateral valuation at the maturity
    boundary. The lender is paid with a portion of the collateral, and the
    remaining collateral is returned to the borrower.
   \item\label{act:liquidate-loan} (\textsc{pcp} 6) Liquidate the collateral in an ongoing loan if the ratio of collateral value to accrued debt falls below $c$. Liquidation is only available up to maturity.
   The lender receives only the amount of collateral whose value equals the accrued debt at the liquidation time, capped by the collateral available in the vault. Any remaining collateral is returned to the borrower.
\end{enumerate}

\begin{remark}[Additional actions]
    Although the set of actions that we introduce can be qualified as bare-bones, it is sufficient to derive strong performance guarantees in the following sections. In practice, participants would require additional actions, such as
allowing the borrower to draw down from the collateral when its value grows quickly; allow the lender to use a loan on secondary markets like US Treasury bills; and allow participants to modify an offer instead of canceling it.
\end{remark}

The fundamental object of our lending protocol is an offer or loan record
involving a lender wallet, a borrower wallet, and a protocol-created
vault wallet. For ease of reference, the notation used in the offer and loan below is summarized in Table~\ref{tab:TableOfNotation}.
An offer to lend created by lender wallet $w_L$ is represented by
\begin{align*}
    \ell
    \triangleq
    (w_L,\emptyset,w_V,\infty,v,0,d,r,\rho,c).
\end{align*}
Here, $w_L$ is the lender wallet, $\emptyset$ is a placeholder for the
borrower wallet, $w_V$ is the vault wallet created for the offer,
$\infty$ is a placeholder for the loan start time, $v$ is the amount of
\textsc{liq} offered for lending, and $0$ is a placeholder for the
amount of \textsc{col} collateral. The parameters $d$, $r$, $\rho$, and
$c$ denote the loan duration, the interest rate per time step, the
required initial collateral ratio, and the liquidation threshold,
respectively. The vault wallet $w_V$ holds $v$ units of \textsc{liq}.

\begin{remark}[Initial collateral ratio $\rho$ and liquidation threshold $c$]
The parameter $\rho$ is the minimum required initial
collateral-to-principal ratio $\phi_{t_0}(z)/v$, whereas $c$ is
the liquidation threshold for the collateral-value-to-accrued-debt
ratio during the loan. For an admissible loan, $\rho>c$ ensures
that the loan starts strictly above the liquidation threshold. Throughout this paper, we focus on permissionless lending and
restrict attention to $c>1$. This provides an overcollateralization
buffer: liquidation can become available while the collateral
value still exceeds the accrued debt. By contrast, thresholds $c\leq1$
provide no such buffer.
\end{remark}

Similarly, an offer to borrow created by borrower wallet $w_B$ is
represented by

\begin{align*}
    b
    \triangleq
    (\emptyset,w_B,w_V,\infty,v,z,d,r,\rho,c).
\end{align*}
Here, $\emptyset$ is a placeholder for the lender wallet, $w_B$ is the borrower wallet, $v$ is the exact amount of \textsc{liq} requested by the
borrower and $z$ is the amount of \textsc{col} collateral. All other entries have
the same meanings as in the offer-to-lend record. The vault wallet $w_V$
holds $z$ units of \textsc{col}.
Finally, an ongoing loan contract between lender wallet $w_L$ and
borrower wallet $w_B$ is represented by
\begin{align*}
    o
    \triangleq
    (w_L,w_B,w_V,t_0,v,z,d,r,\rho,c).
\end{align*}
Here, $t_0$ is the loan start time, and all other entries have the same
meanings as in the offer records. The vault wallet $w_V$ holds $z$ units
of \textsc{col}. The maturity boundary of the ongoing loan is
    $t_0+d$.

\begin{definition}[Accrued debt and health factor]
\label{def:health_factor}
For an ongoing loan $o=(w_L,w_B,w_V,t_0,v,z,d,r,\rho,c)$
with $v>0$, $c>0$, and $1+r>0$, its accrued debt and health factor
at time $t$ are respectively:
\begin{align}
    D_t
    &\triangleq (1+r)^{t-t_0}v,
    \label{eq:accrued_debt_definition}\\
    h_t
    &\triangleq \frac{\phi_t(z)}{cD_t}
    = \frac{\phi_t(z)}{c(1+r)^{t-t_0}v}.
    \label{eq:health_factor_definition}
\end{align}
Here, $t_0$ is the loan start time, $t$ is the current time,
$z$ is the amount of \textsc{col} collateral held in the vault,
$\phi_t(z)$ is the USD value of this collateral at time $t$,
$v$ is the principal value in USD, $r$ is the interest rate
per time step, and $c$ is the liquidation threshold.
\end{definition}

\begin{remark}[Adaptive health factor]
The health factor in \eqref{eq:health_factor_definition} follows
the convention used in Aave \cite{aave} and Morpho \cite{morpho},
incorporating the liquidation threshold directly into the ratio.
Our parameter $c>1$ corresponds to a conventional liquidation
threshold $1/c<1$. Liquidation becomes available when $h_t<1$,
or equivalently when $\phi_t(z)/D_t<c$.
Unlike overly conservative formulations that compare the collateral value with the
loan's final maturity value $(1+r)^dv$ \cite{hps}, our definition is adaptive and 
compares against the accrued debt $D_t=(1+r)^{t-t_0}v$ up to time $t$.
\end{remark}

\begin{definition}[Admissible loan at origination]
\label{def:admissible_loan}
A loan that starts at time $t_0$, with principal $v$, collateral $z$,
initial collateral ratio $\rho$, and liquidation threshold $c$, is admissible at origination if
\begin{align*}
    \rho>c>0
    \qquad\mbox{and}\qquad
    \phi_{t_0}(z)\geq\rho v.
\end{align*}
\end{definition}
At origination time, admissibility ensures that
\begin{align*}
    h_{t_0}
    =
    \frac{\phi_{t_0}(z)}{cv}
    \geq
    \frac{\rho}{c}
    >
    1,
\end{align*}
so that the loan is not immediately eligible for
liquidation.

\begin{remark}[Offer creation and take-up]
An offer may be created under Action~\ref{act:offer-to-lend} or~\ref{act:offer-to-borrow} only if its specified parameters satisfy the applicable conditions in Definition~\ref{def:admissible_loan}. At take-up under Action~\ref{act:take-up-offer}, the resulting loan must satisfy that definition under the current valuation $\phi_t$.
For an offer to borrow, collateral $z$ is committed at creation. A fall in its value may prevent take-up for a proposed loan amount $v$, while the offer remains recorded. The borrower may cancel it through Action~\ref{act:cancel-offer} and submit a new offer with sufficient collateral through Action~\ref{act:offer-to-borrow}. If left open, the offer may become executable again for that amount if its collateral value recovers. For an offer to lend, no collateral amount is fixed at creation; the borrower chooses and supplies sufficient collateral at take-up.
\end{remark}

\subsection{State Transitions}\label{sec:transition}

In contrast to lending pools, our state space is much more complicated. It is not even a vector space or finite-dimensional metric space. It is composed of all possible sets of offers, where each offer can be either an offer for borrowing or an offer for lending.
A state in our protocol is therefore an element from a countably infinite-dimensional state space, with the form
\begin{align*}
    S_t \triangleq 
    (
    \underbrace{\ell^{(1)}, \ldots, \ell^{(m_t)}}_{\textrm{offers to lend}},
    \underbrace{b^{(1)}, \ldots, b^{(n_t)}}_{\textrm{offers to borrow}},
    \underbrace{o^{(1)}, \ldots, o^{(k_t)}}_{\textrm{ongoing loans}},
    W_t^{(1)}, \ldots, W_t^{(N_t)},
    \phi_t
    ),
\end{align*}
where \(m_t\) and \(n_t\) denote the numbers of outstanding offers available on the market, \(k_t\) is the number of ongoing loan contracts, and \(W_t^{(1)},\ldots,W_t^{(N_t)}\) denote the balances of all wallets in the protocol, including lender wallets, borrower wallets, and vault wallets. Here \(W_t^{(1)},\ldots,W_t^{(N_t)}\) include ordinary participant wallets as well as vault wallets. The role of a wallet in a particular offer or loan is specified by the position in which it appears in $\ell,b,o$.

The three collections of records in the state are disjoint at each time \(t\). A record in the first two collections is an open offer, whereas a record in the third collection is an ongoing loan. When an open offer is taken up, it is removed from the corresponding offer collection and inserted into the ongoing-loan collection.
The initial state is an empty market with no offers, and all wallet balances are at their initial values.

The state transitions keep track of the available offers on the market, the ongoing contracts, wallet balances, vault balances, and the valuation of collateral.

Each action taken by a participant (wallet) at time \(t\) triggers a corresponding state transition as follows.
\begin{enumerate}
\item Wallet \(w_L\) initiates the offer to loan \(\ell = (w_L, \emptyset, w_V, \infty, v, 0, d, r, \rho, c)\):
\begin{align*}
    S_{t+1} &= S_t \cup \ell,\\
    W_{L,t+1}[1] &= W_{L,t}[1] - v \quad (\mbox{commit \textsc{liq}}),\\
    W_{V,t+1}[1] &= W_{V,t}[1] + v \quad (\mbox{vault receives \textsc{liq}}).
\end{align*}
The vault wallet \(w_V\) receives \(v\) in \textsc{liq}.

\item Borrower wallet \(w_B\) takes up an offer to lend \(\ell\) from lender wallet \(w_L\) by committing \(z\) units of \textsc{col} at time \(t_0\):
\begin{align*}
    S_t = (\ldots, \ell, \ldots) &\mapsto S_{t+1} = (\ldots, o, \ldots),\\
    W_{B,t+1}[1] &= W_{B,t}[1] + v, \quad (\mbox{receive \textsc{liq}})\\
    W_{B,t+1}[0] &= W_{B,t}[0] - z, \quad (\mbox{collateral})\\
    W_{V,t+1}[1] &= W_{V,t}[1] - v, \quad (\mbox{vault sends \textsc{liq}})\\
    W_{V,t+1}[0] &= W_{V,t}[0] + z \quad (\mbox{vault receives \textsc{col}}).
\end{align*}
The vault wallet \(w_V\) sends \(v\) in \textsc{liq} to the borrower wallet and receives \(z\) in \textsc{col}.
The borrower wallet \(w_B\) is added to the smart contract, along with the current time \(t\), and the collateral \(z\):
\begin{align*}
    \ell = (w_L,\emptyset,w_V,\infty,v,0,d,r,\rho,c)
    \mapsto
    (w_L,w_B,w_V,t,v,z,d,r,\rho,c)=o.
\end{align*}
This transaction is only allowed if the collateral \(z\) amount results in an admissible loan (Definition~\ref{def:admissible_loan}). In practice, the collateral amount is typically chosen so that its initial value is determined by the collateral ratio $\rho$, with $\phi_{t_0}(z)=\rho v$ up to implementation-level rounding and oracle-update effects.

It is extremely important that the borrower calculates the value \(z\) carefully to avoid falling into a liquidation condition once the loan starts, which is subject to the valuation process and the volatility in the value of \textsc{col} relative to \textsc{liq}.

\item Borrower $w_B$ initiates an offer to borrow $b$, requesting $v$ units of \textsc{liq} and committing $z$ units of \textsc{col} as collateral:
\begin{align*}
    S_{t+1} &= S_t \cup b,\\
    W_{B,t+1}[0] &= W_{B,t}[0] - z \quad (\mbox{commit \textsc{col} as collateral}),\\
    W_{V,t+1}[0] &= W_{V,t}[0] + z \quad (\mbox{vault receives \textsc{col}}).
\end{align*}

\item Lender wallet $w_L$ takes up an offer to borrow $b$ at time $t$ by providing the principal $v$ specified in the offer:
\begin{align*}
    S_t = (\ldots, b, \ldots) &\mapsto S_{t+1} = (\ldots, o, \ldots),\\
    W_{L,t+1}[1] &= W_{L,t}[1] - v, \quad (\mbox{lender commits \textsc{liq}})\\
    W_{B,t+1}[1] &= W_{B,t}[1] + v, \quad (\mbox{borrower receives \textsc{liq}}),\\
    W_{V,t+1}[0] &= W_{V,t}[0] \quad (\mbox{vault keeps \textsc{col} collateral}).
\end{align*}
The lender wallet $w_L$ is added to the smart contract, along with the current time $t$:
\begin{align*}
    b = (\emptyset,w_B,w_V,\infty,v,z,d,r,\rho,c)
    \mapsto
    (w_L,w_B,w_V,t,v,z,d,r,\rho,c)=o.
\end{align*}

This transaction is only allowed if the collateral amount $z$ and the loan value $v$ result in an admissible loan 
(Definition~\ref{def:admissible_loan}), thereby satisfying the minimum required initial collateral ratio $\rho$.

\item The offer creator cancels that offer to lend \(\ell\):
\begin{align*}
    S_{t+1} &= S_t \setminus \ell,\\
    W_{V,t+1}[1] &= W_{V,t}[1] - v \quad (\mbox{vault sends \textsc{liq}}),\\
    W_{L,t+1}[1] &= W_{L,t}[1] + v \quad (\mbox{redeem \textsc{liq}}).
\end{align*}

\item The offer creator cancels an offer to borrow \(b\):
\begin{align*}
    S_{t+1} &= S_t \setminus b,\\
    W_{V,t+1}[0] &= W_{V,t}[0] - z \quad (\mbox{vault sends \textsc{col}}),\\
    W_{B,t+1}[0] &= W_{B,t}[0] + z \quad (\mbox{redeem \textsc{col}}).
\end{align*}

\item Borrower \(w_B\) tops up collateral by \(\Delta z\) units of \textsc{col}:
\begin{align*}
    o = (w_L,w_B,w_V,t_0,v,z,d,r,\rho,c)
    &\mapsto
    o' = (w_L,w_B,w_V,t_0,v,z+\Delta z,d,r,\rho,c),\\
    W_{B,t+1}[0] &= W_{B,t}[0] - \Delta z \quad (\mbox{commit \textsc{col} as collateral}),\\
    W_{V,t+1}[0] &= W_{V,t}[0] + \Delta z \quad (\mbox{vault receives \textsc{col}}).
\end{align*}

\item Interrupt loan: the borrower $w_B$ ends the loan early at time $t$ such that $t-t_0 \le d$:
\begin{align*}
    S_{t+1} &= S_t \setminus o,\\
    W_{B,t+1}[1] &= W_{B,t}[1] - (1+r)^d \cdot v, \quad (\mbox{pay back loan with interest in \textsc{liq}}),\\
    W_{L,t+1}[1] &= W_{L,t}[1] + (1+r)^d \cdot v, \quad (\mbox{lender receives \textsc{liq}}),\\
    W_{V,t+1}[0] &= W_{V,t}[0] - z \quad (\mbox{vault returns \textsc{col} collateral}),\\
    W_{B,t+1}[0] &= W_{B,t}[0] + z \quad (\mbox{return full \textsc{col} collateral}).
\end{align*}

\item The lender liquidates loan $o$ at checkpoint $t\in\{t_0+1,\ldots,t_0+d\}$ when the collateral value falls below the liquidation threshold, or equivalently when $h_t<1$.
The liquidation action terminates the loan and distributes the collateral between the lender and borrower respectively as follows:
    \begin{align}
        a_t \triangleq \min\{\phi_t^{-1}(D_t),z\},\quad        e_t \triangleq z-a_t. \label{eq:settlement_lender_collateral}
    \end{align}
    The state transition is
    \begin{align}
        S_{t+1} &= S_t\setminus o, \label{eq:settlement_first_state_transition}\\
        W_{V,t+1}[0] &= W_{V,t}[0]-z
        \quad (\mbox{vault sends \textsc{col} collateral}),\nonumber\\
        W_{L,t+1}[0] &= W_{L,t}[0]+a_t
        \quad (\mbox{lender receives }a_t\mbox{ units of \textsc{col}}),\nonumber\\
        W_{B,t+1}[0] &= W_{B,t}[0]+e_t
        \quad (\mbox{return remaining \textsc{col} collateral}). \label{eq:settlement_last_state_transition}
    \end{align}

\item Natural end of loan $o$: at the maturity time $t=t_0+d$, if the loan is still ongoing, the protocol first evaluates the liquidation condition. If $h_t<1$, the liquidation action described above applies. Otherwise, the protocol service provider immediately executes the natural-ending settlement. The settlement follows the same collateral-distribution rules as benevolent liquidation. Specifically, the collateral amounts are given by \eqref{eq:settlement_lender_collateral}, and the state transition is given by \eqref{eq:settlement_first_state_transition}--\eqref{eq:settlement_last_state_transition}, all evaluated at $t=t_0+d$.
To avoid the complexity of performing an exchange between \textsc{col} and \textsc{liq}, the repayment occurs in units of \textsc{col} as opposed to \textsc{liq}, and lender has the burden of converting it back to $(1+r)^d v$ units \textsc{liq}.
\end{enumerate}

\begin{remark}[Asset-agnostic lender]
    Observe that the interrupt loan ending repays the lender in \textsc{liq}, whereas the other endings repay in \textsc{col}. We assume throughout the paper that the lender is agnostic to receiving the repayment in \textsc{col} or \textsc{liq}. In other words, we ignore costs for converting between \textsc{col} and \textsc{liq}.
\end{remark}

\begin{remark}[Alternative liquidation rules]\label{rem:full}
    The above protocol presents a rule for distributing the collateral that prioritizes repayment of the lender. A more general protocol could conceivably allow a distinct distribution rule for each loan. For instance, a distribution rule that transfers the full collateral to the lender in the event of a liquidation would encourage lenders to participate and borrowers to keep the health factor at more conservative levels.
\end{remark}

\begin{figure}[htbp]
\centering
\begin{tikzpicture}[
    >=stealth',
    font=\small,
    loanstate/.style={
        circle,
        draw=black,
        thick,
        fill=gray!15,
        minimum size=11mm,
        inner sep=0pt
    }
]
    \node[loanstate,fill=blue!20]
        (start) at (0,2.8) {Start};

    \node[loanstate] (S1) at (-4.5,0) {S1};
    \node[loanstate] (S2) at (4.5,0) {S2};
    \node[loanstate] (S3) at (0,0) {S3};
    \node[loanstate,double,double distance=1pt]
        (S4) at (0,-3.5) {S4};

    \draw[->]
        (start) -- node[above,sloped]
        {create offer to borrow} (S1);

    \draw[->]
        (start) -- node[above,sloped]
        {create offer to lend} (S2);

    \draw[->]
        (S1) -- node[above] {take up} (S3);

    \draw[->]
        (S2) -- node[above] {take up} (S3);

    \draw[->]
        (S1) -- node[left] {cancel} (S4);

    \draw[->]
        (S2) -- node[right] {cancel} (S4);

    \draw[->]
        (S3) -- node[
            right,
            pos=0.4,
            align=left
        ] {interrupt,\\liquidate,\\natural end} (S4);

    \path[->]
        (S3) edge[loop above]
        node[above] {top up} (S3);
\end{tikzpicture}
\caption{Relation between offers and loans.}
\label{fig:placeholder}
\end{figure}
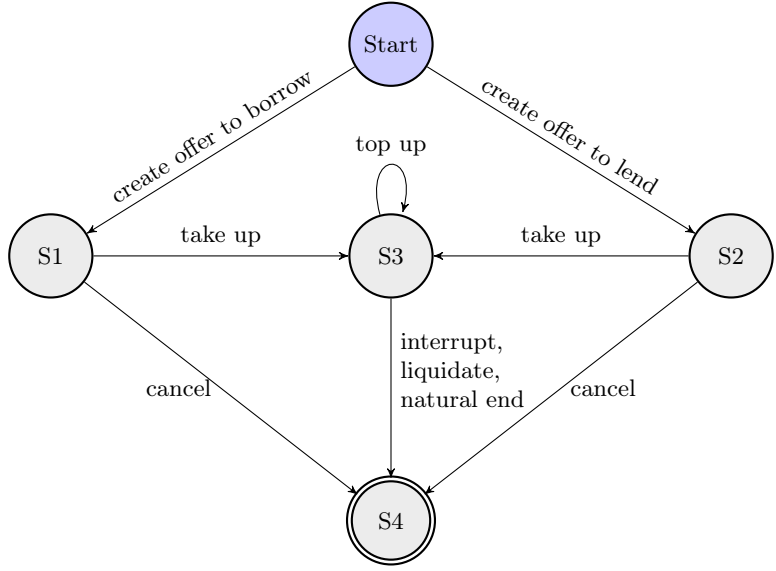

Figure~\ref{fig:placeholder} illustrates the transition from an offer to an ongoing loan and to a terminated loan.
The state $S1$ denotes an available offer to borrow, $S2$ denotes an available offer to lend, $S3$ denotes an ongoing loan, and $S4$ denotes a terminated loan or canceled offer.
The arrows show the transitions following the different actions of the protocol.

\begin{figure}[htbp]
    \centering
    \captionsetup{margin=5mm,font=small}
    \includegraphics[width=\textwidth]{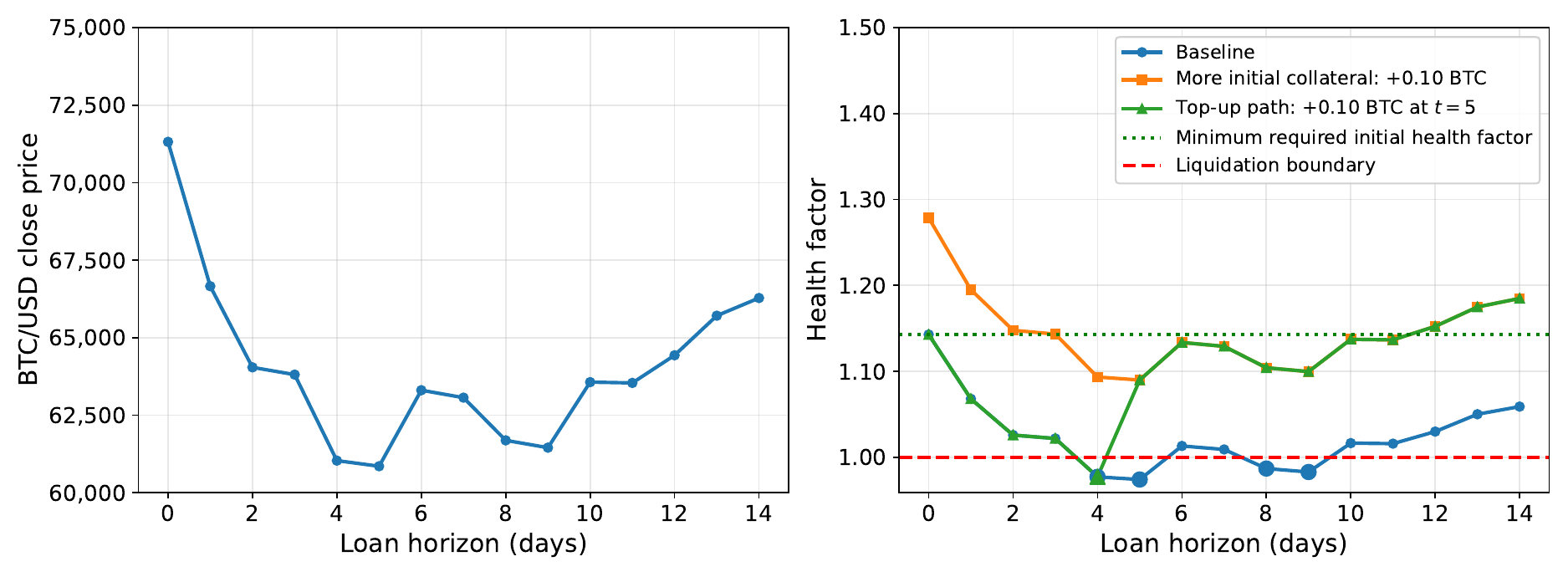}
    \caption{BTC/USD daily close prices and health factor paths under three collateral strategies.}
    \label{fig:btc_price_health_factor}
\end{figure}

Figure~\ref{fig:btc_price_health_factor} illustrates the evolution of the health factor over the course of a loan that starts at time $0$.
The left panel shows BTC-USD daily close prices $P_t$ from 2026-06-01 to 2026-06-15. The right panel compares three collateral strategies for a loan with $v=50000$, $c=1.05$, $\rho=1.2$, and an APR of $8.00\%$. Health factors are evaluated using \eqref{eq:health_factor_definition}, with $\phi_t(z)=P_tz$: higher BTC prices or additional collateral increase the health factor, while accrued debt growth reduces it. The dotted and dashed horizontal lines mark the minimum required initial health factor and the liquidation boundary, respectively.

The baseline path (circle markers) starts with exactly the required collateral, giving an initial health factor of approximately $1.14$. The more-initial-collateral path (square markers) adds $0.1$ BTC at origination, increasing the initial health factor to approximately $1.28$, whereas the top-up path (triangle markers) adds the same amount of BTC on day $5$. Before day $5$, the top-up path follows the baseline below the more-initial-collateral path; from day $5$ onward, it coincides with the latter because both hold the same amount of BTC.
Both the baseline and top-up paths first fall below the liquidation boundary on day $4$ (2026-06-05). The baseline is below the boundary on days $4$, $5$, $8$, and $9$, whereas the top-up path is below it only on day $4$. Note that the continuations of the baseline and top-up paths beyond day $4$ are shown for illustration only. If liquidation is executed at the first eligible checkpoint, both loans terminate on day $4$, before the scheduled top-up. By contrast, the more-initial-collateral path remains above the boundary throughout this historical window, with a minimum health factor of approximately $1.09$, avoiding liquidation along this price path.

\begin{remark}[Regime of few participants]
    When there are few participants, we can restrict the values that the parameters can take in order to increase the likelihood that offers are taken up. For example, \(d,\rho,c\) take finite values, and \(v\) is in multiples of \(100\).
\end{remark}

\begin{remark}[Frontrunning]
    Frontrunning can occur in various ways. For example, a lender cancels an existing offer to lend and then is about to take up an offer to borrow, but another lender swoops in and preempts this take-up action by paying more gas fee.
\end{remark}

\section{Performance Guarantees}
\label{sec:guarantee}

In this section, we derive two guarantees for participants in our protocol.
First, we quantify the probability of liquidation of the collateral during the contractual loan period. Second, we quantify the probability that the lender  receives less than the
promised repayment when a loan terminates.
These are arguably the most important quantities for borrowers and lenders, and can serve as stepping stones to derive other guarantees.
Both guarantees are derived under a pessimistic scenario where the borrower never tops up the collateral, and where the lender always triggers the liquidation action as soon as it is allowed.

Consider an admissible loan and measure time relative to loan
origination. Throughout this section, loan states are indexed by
$t\in\{0,\ldots,d\}$, where $t=0$ is the origination time and $t=d$ is
the maturity time. We assume that neither a
collateral top-up nor an interrupt action occurs after origination; hence, the collateral amount $z$ remains fixed. 
As specified in Section~\ref{sec:transition}, the loan is liquidated at the first checkpoint $t=1,\ldots,d$ where $h_t<1$. At maturity time $d$,
liquidation check is performed one last time before the natural-ending settlement occurs.

We model the collateral valuation as a geometric random walk. This is the
discrete-time analogue of the geometric Brownian-motion models commonly
used in finance \cite{hlz,merton}. Because the collateral amount $z$
remains fixed throughout the analysis, it is sufficient to specify the
valuation process for the collateral held by the particular loan under
consideration.

\begin{assumption}[Gaussian log-valuation increments]
\label{as:1}
For the collateral amount $z$ held by the loan,
\begin{align*}
    \phi_t(z)
    =
    e^{\epsilon_t}\phi_{t-1}(z),
    \qquad
    t=1,\ldots,d,
\end{align*}
where, conditional on the information observed at origination,
$\epsilon_1,\ldots,\epsilon_d$ are independent and identically distributed
as $N(\mu,\sigma^2)$, with $\mu\in\mathbb R$ and $\sigma^2>0$.
\end{assumption}

With time measured from loan origination, so that $t_0=0$, define the
log inverse health factor by
\begin{align}
    Y_t
    \triangleq
    -\log h_t
    =
    \log\frac{cD_t}{\phi_t(z)},
    \qquad
    t=0,\ldots,d.
    \label{eq:log_inverse_health_factor_gaussian}
\end{align}
For the admissible loan under consideration, let $Y_0=y_0$ denote the observed value of the log inverse health factor at origination. For the remainder of the paper, all distributional statements,
probabilities, and expectations are conditional on the observed
initial value $Y_0=y_0$.
By Assumption~\ref{as:1} and 
\eqref{eq:log_inverse_health_factor_gaussian}, and defining the increments $\Delta Y_k \triangleq \log(1+r)-\epsilon_k$, we obtain
\begin{align}
    Y_t
    =
    y_0+\sum_{k=1}^t \underbrace{\left(\log(1+r)-\epsilon_k\right)}_{\Delta Y_k},
    \qquad
    t=1,\ldots,d.
    \label{eq:Y_gaussian_representation}
\end{align}
Starting from the
initial position $y_0$, each time increment $k$ brings an additional debt term $\log(1+r)$, which is countered by a collateral appreciation term $\epsilon_k$. 

We define the liquidation time in terms of the sequence $Y_t$ by observing that liquidation is triggered if $h_t < 1 \iff Y_t >0$.

\begin{definition}[Liquidation time]
    Let $T$ denote the time when liquidation is triggered, with the convention that it takes value $\infty$ if no liquidation is triggered. Formally,
    \begin{align}
    T
    \triangleq
    \begin{cases}
    \infty, &\text{if $Y_1\leq0,\ldots,Y_d\leq0$},\\
    \min\left\{
    t\in\{1,\ldots,d\}:Y_t>0
    \right\},
    &
    \text{otherwise}
    .
    \end{cases}
    \label{eq:decollateralization_time}
\end{align}
\end{definition}
In turn, we have
\begin{align}
\{T=1\} = \{Y_1>0\}, \quad 
    \{T=t\}
    =
    \left\{
    Y_1\leq0,\ldots,Y_{t-1}\leq0,\;
    Y_t>0
    \right\} \quad t=2,\ldots,d.
    \label{eq:first_liquidation_event}
\end{align}
The event $T=\infty$ denotes a natural-ending settlement at time $d$.

\subsection{Borrower Liquidation Probability}
\label{subsec:borrower_liquidation_probability}

Since the liquidation event is entirely determined by the
path $Y_1,\ldots,Y_d$, in this section, first establish the joint distribution of this sequence and then quantify the probability of liquidation.

\begin{lemma}[Multivariate Gaussian representation]
\label{prop:multivariate_gaussian_representation}
Under Assumption~\ref{as:1}, for every
$t\in\{1,\ldots,d\}$, the vector random variable $(Y_1,\ldots,Y_t)$
has a Gaussian distribution $N_t\left(\boldsymbol{\mu}_t^Y,\boldsymbol{\Sigma}_t^Y\right)$ with mean and covariance matrices:
\begin{align}
    (\boldsymbol{\mu}_t^Y)_i
    &\triangleq
    y_0+i\bigl(\log(1+r)-\mu\bigr),
    \quad
    (\boldsymbol{\Sigma}_t^Y)_{ij}
    \triangleq
    \sigma^2\min\{i,j\},
    \qquad \text{for } i,j \in 1,\ldots,t.\label{eq:816}
\end{align}
\end{lemma}
We denote the multivariable Gaussian cumulative distribution function of $\mathbf G\sim N_t(\boldsymbol{\mu},\Sigma)$ by
    $\Phi_t(\mathbf x;\boldsymbol{\mu},\Sigma)
    \triangleq
    \mathbb P(\mathbf G\leq\mathbf x)$.
Let $\mathbf 0_t$ denote the $t$-dimensional zero vector. Our main guarantee for the borrower is the following probability of liquidation.

\begin{theorem}[Borrower liquidation probability]
\label{thm:borrower_liquidation_probability}
Under Assumption~\ref{as:1} and the loan setup in
Section~\ref{sec:guarantee}, the probability of liquidation at time $t$ is
\begin{align}
    \mathbb P(T=t) =
    \begin{cases}
    1-\Phi_1\left(0;\boldsymbol{\mu}_1^Y,\boldsymbol{\Sigma}_1^Y \right), & \text{for } t = 1 \\
        \Phi_{t-1}\left(\mathbf 0_{t-1};\boldsymbol{\mu}_{t-1}^Y,\boldsymbol{\Sigma}_{t-1}^Y\right)
    -
    \Phi_t\left(\mathbf 0_t;\boldsymbol{\mu}_t^Y,\boldsymbol{\Sigma}_t^Y\right),  & \text{for } t\in\{2,\ldots,d\},
\end{cases}
\end{align}
and the probability of liquidation by time $s$ is
\begin{align}
    \mathbb P(T\leq s)
    =
    1-
    \Phi_s\left(
    \mathbf 0_s;
    \boldsymbol{\mu}_s^Y,
    \boldsymbol{\Sigma}_s^Y
    \right), \quad \text{for }s\in\{1,\ldots,d\}.\label{eq:exact_borrower_liquidation_probability}
\end{align}
\end{theorem}

\begin{remark}[Operating regime]
    This theorem shows that, in order to avoid frequent liquidation events or requiring frequent top ups, the lending protocol requires that the value of the collateral increases geometrically faster than the value of the loan with accumulated interest, i.e., that $\mu > \log(1+r)$ in \eqref{eq:816}.
\end{remark}

\begin{remark}[Borrower control]
The borrower can use
Theorem~\ref{thm:borrower_liquidation_probability} to select loan
terms that keep the liquidation probability below
a desired level. Subject to the admissibility
requirements, a smaller loan amount $v$, shorter duration $d$, lower
interest rate $r$, lower liquidation threshold $c$, or larger collateral
amount $z$ reduces liquidation probability. During the course of the loan, a re-evaluated version of the probability in \eqref{eq:exact_borrower_liquidation_probability} can also guide whether
and how much to top up the collateral.
\end{remark}

\subsection{Lender repayment shortfall Probability}
\label{subsec:lender_repayment_shortfall_probability}

In this section, we give a guarantee on the probability that the lender suffers a repayment shortfall when the loan terminates. Again, for simplicity, we assume that the borrower lets the loan terminate never topping it up nor interrupting it. We will see that the lender can control this probability through the parameters $\rho$ and $c$.

\begin{definition}[Repayment shortfall]\label{def:repayment_shortfall}
    We say that the lender suffers a repayment shortfall by time $s\in\{1,\ldots,d\}$ if the collateral cannot cover the accrued debt, i.e., $T\leq s$ and $\phi_T(z)<D_T$.
\end{definition}
Observe that no shortfall can occur if the borrower interrupts the loan or if $T = \infty$ (no liquidation).
In the event of a liquidation at time $t$, if the lender receives a portion of the collateral equivalent to the accrued debt $D_t$, it incurs a loss due to missing out on interest from time $t$ to the initial maturity $d$, but does not suffer a shortfall.
For simplicity, we ignore the slippage and other costs that the lender may incur to convert the received portion of \textsc{col}
to \textsc{liq}.

\begin{theorem}[Lender repayment shortfall probability]
\label{thm:lender_repayment_shortfall_probability}
Under Assumption~\ref{as:1}, with
the loan setup in Section~\ref{sec:guarantee}, the lender
repayment shortfall probability by time $s\in\{1,\ldots,d\}$ is
\begin{align}
&\mathbb P(T\leq s,\ \phi_T(z)<D_T)\nonumber\\
&=
1-\Phi_1\left(
    \log c;\boldsymbol{\mu}_1^Y,\boldsymbol{\Sigma}_1^Y
\right)
+\sum_{t=2}^{s}
\left[
    \Phi_{t-1}\left(
        \mathbf 0_{t-1};
        \boldsymbol{\mu}_{t-1}^Y,\boldsymbol{\Sigma}_{t-1}^Y
    \right)
    -
    \Phi_t\left(
        (\mathbf 0_{t-1}^{\top},\log c)^{\top};
        \boldsymbol{\mu}_t^Y,\boldsymbol{\Sigma}_t^Y
    \right)
\right],
\label{eq:exact_lender_repayment_shortfall_probability_revised}
\end{align}
with the sum interpreted as zero when $s=1$.
\end{theorem}

\begin{remark}[Lender control]
For a fixed admissible collateral amount $z$ set by the borrower, the value of $c$ and $\rho$ jointly control the value $y_0 = \log \frac{cv}{\phi_0(z)}$, subject to the admissibility constraint $c < \rho$.
In turn, the value $y_0$ controls $\boldsymbol{\mu}_t^Y$, which finally controls the probability of \eqref{eq:exact_lender_repayment_shortfall_probability_revised}.
Therefore, the lender can control this probability by setting $c$ and $\rho$ appropriately or by selecting an offer with such values.
\end{remark}

\section{Generalized Liquidation and Repayment Shortfall Probabilities}
\label{sec:non_gaussian_liquidation}
The preceding analysis derives exact liquidation and repayment shortfall probabilities under Gaussian log-valuation increments $\epsilon_t$. We now employ characteristic functions to generalize those results to arbitrary increments that are independent and identically distributed. 
This generalization most importantly allows us to select a suitable family of distributions based on prior knowledge about the valuation dynamics of the collateral.
The tradeoff is that the general results require numerical approximation of the corresponding integrals.


As before, we consider an ongoing loan that started at time $t_0=0$.
We take the observed origination value of $Y_0$ to be $y_0<0$, and recall that
\begin{align*}
    Y_t = y_0 + \sum_{k=1}^t \Delta Y_k, \quad \text{for }t=1,\ldots,d.
\end{align*}
We assume that the log-valuation increments
$\epsilon_1,\ldots,\epsilon_d$ are i.i.d., and denote the common
characteristic function by
$\varphi_\epsilon(\omega)\triangleq
\mathbb E[e^{i\omega\epsilon_1}]$.
In turn, the sequence $\Delta Y_t$ is also i.i.d. with a common characteristic function
\begin{align}
    f(\omega)
    \triangleq \mathbb E[e^{i\omega\Delta Y_1}]
    = e^{i\omega\log(1+r)}\varphi_\epsilon(-\omega).
    \label{eq:debt_adjusted_increment_cf}
\end{align}

Define the running maximum of the changes in $Y_t$ from its origination
value and its characteristic function by
\begin{align*}
    M_s\triangleq\max_{0\leq t\leq s}(Y_t-y_0),
    \qquad
    \chi_s(\omega)\triangleq\mathbb E[e^{i\omega M_s}],
    \qquad s=0,\ldots,d.
\end{align*}
For $t=1,\ldots,d$, also define the positive-part sequence and its characteristic functions:
\begin{align*}
(Y_t-y_0)_+ = \max\{0, Y_t-y_0\}, \quad 
    \psi_t(\omega)\triangleq
    \mathbb E[e^{i\omega(Y_t-y_0)_+}].
\end{align*}

\begin{theorem}[Liquidation probability under i.i.d. returns]
\label{thm:iid_increment_liquidation_probability}
Suppose that $\epsilon_1,\ldots,\epsilon_d$ are i.i.d. and that the distribution function of $M_s$ is continuous at $-y_0$
for every $s\in\{1,\ldots,d\}$. 
Then, the characteristic functions $\psi_t$ and $\chi_s$ are
\begin{align}
    \psi_t(\omega)
    &=
    \frac{1+f(\omega)^t}{2}
    +
    \frac{1}{2\pi i}\int_0^\infty
    \frac{
        [f(\omega+u)^t-f(\omega-u)^t]
        -
        [f(u)^t-f(-u)^t]
    }{u}\,\mathrm{d}u,
    \qquad t=1,\ldots,d.\label{eq:positive_partial_sum_cf_four_term}
\\
    \chi_0(\omega)&=1,
    \qquad
    \chi_s(\omega)
    =\frac1s\sum_{t=1}^{s}
    \psi_t(\omega)\chi_{s-t}(\omega),
    \qquad s=1,\ldots,d,
\label{eq:spitzer_finite_recursion}
\end{align}
Moreover, the liquidation
probability by time $s\in\{1,\ldots,d\}$ is
\begin{align}
    \mathbb P(T\leq s)
    =
    1-\frac{2}{\pi}\int_0^\infty
    \frac{\sin(-y_0\omega)}{\omega} \operatorname{Re}\chi_s(\omega)
    \,\mathrm{d}\omega.
\label{eq:exact_non_gaussian_liquidation_probability}
\end{align}
\end{theorem}


\begin{remark}[Continuity assumption]
For each $M_s$ to have a distribution function continuous at $-y_0$, it is sufficient that each $Y_t$ have no probability mass at $0$ for every
$t=1,\ldots,d$. For certain distributions of the i.i.d. increments $\epsilon_k$,
including Gaussian, gamma, Poisson, etc., their sum is well-known, allowing the distribution of $Y_t$ to be derived analytically. Otherwise,
repeated samples of $\epsilon_1,\ldots,\epsilon_t$ yield samples of $Y_t$,
which can be used to estimate its distribution and assess its continuity at $0$ by inspection.
\end{remark}


The integral in
\eqref{eq:exact_non_gaussian_liquidation_probability}
may appear similar to
\cite[Eq.~3.6]{abate1992fourier}
and \cite[Eq.~1]{nuttall1969numerical},
but the characteristic function $\chi_s$ in
\eqref{eq:spitzer_finite_recursion}
is computed using an approach similar to that of
Haslip and Kaishev \cite[Proposition~3.4]{haslip2014lookback},
who adopt {\"O}hgren's recursion \cite{ohgren2001remark}
and Spitzer's Identity.
Whereas Theorem~\ref{thm:borrower_liquidation_probability}
applies only to the Gaussian family of distributions,
Theorem~\ref{thm:iid_increment_liquidation_probability}
applies to every possible family.
This broader applicability however requires the numerical evaluation
of two integrals (\eqref{eq:positive_partial_sum_cf_four_term}
and \eqref{eq:exact_non_gaussian_liquidation_probability})
over unbounded domains and without analytical form solutions.
This introduces both truncation errors and quadrature errors
from discretizing the integrals.
These errors can be reduced by increasing the truncation
cutoffs and reducing the step sizes
\cite{abate1992fourier,nuttall1969numerical}.
A detailed analysis of these numerical errors is beyond the scope
of this paper.

To give the corresponding guarantee for lenders, we first define the sequence of joint characteristic functions
    \begin{align}
        \Gamma_s(u,\omega)
        &\triangleq
        \mathbb E[e^{iuM_s+i\omega(Y_s-y_0)}],
        \qquad (u,\omega)\in\mathbb R^2,\quad s=0,\ldots,d.
        \label{eq:lender_joint_cf}
    \end{align}
        
\begin{theorem}[Repayment shortfall probability under i.i.d. returns]
    \label{thm:iid_lender_shortfall_exact} Suppose that $\epsilon_1,\ldots,\epsilon_d$ are i.i.d. and that
    the distribution function of $M_t$ is continuous at both $-y_0$
    and $\log c-y_0$ for every $t\in\{1,\ldots,d\}$.
    Let $f$, $\psi_t$, and $\chi_t$ be the characteristic functions
    introduced before Theorem~\ref{thm:iid_increment_liquidation_probability}.
Then, the characteristic functions $\Gamma_s$ are
    \begin{align}
        \Gamma_0(u,\omega)
        &=1,
        \nonumber
        \\
        \Gamma_s(u,\omega)
        &=\frac{1}{s}\sum_{t=1}^{s}
        \bigl[\psi_t(u+\omega)+f(\omega)^t-\psi_t(\omega)\bigr]
        \Gamma_{s-t}(u,\omega),
        \qquad s=1,\ldots,d.
        \label{eq:lender_joint_recursion}
    \end{align}
    Moreover, for $s=1,\ldots,d$, the repayment shortfall
    probability is
    \begin{align}
&\mathbb P(T\le s,\ \phi_T(z)<D_T)
=\sum_{t=1}^{s}\Bigg[
    \frac{1}{\pi}\int_{0}^{\infty}
    \frac{\sin(-y_0 \xi)}{\xi}\,\operatorname{Re}\chi_{t-1}(\xi)
    \,d\xi
\nonumber\\
&+\frac{1}{\pi^2}\int_{0}^{\infty}
    \operatorname{Im}\Biggl\{
    e^{-i\omega(\log c-y_0)}\frac{f(\omega)}{\omega}
    \int_{0}^{\infty}
    \frac{\sin(-y_0u)}{u}
    \bigl[
        \Gamma_{t-1}(u,\omega)
        +\Gamma_{t-1}(-u,\omega)
    \bigr]\,du
    \Biggr\}d\omega
    \Bigg].\label{eq:lender_exact_probability_im}
\end{align}
\end{theorem}

\subsection{Empirical Liquidation and Repayment Shortfall Probabilities}
\label{subsec:empirical_liquidation_probability}

To employ Theorems~\ref{thm:iid_increment_liquidation_probability} and~\ref{thm:iid_lender_shortfall_exact}, we need to know the distribution or characteristic function $\varphi_\epsilon$ for the sequence $\epsilon_t$.
In this section, we consider the inference problem where this distribution is unknown and must be inferred from historical samples. The first step is to select the most suitable one from a set of three candidate parametric distributions by estimating the corresponding parameters. The second step is to calculate the liquidation probability and repayment shortfall probabilities by plugging-in this distribution into
Theorems~\ref{thm:iid_increment_liquidation_probability} and~\ref{thm:iid_lender_shortfall_exact}, respectively.

To illustrate these two steps, we use daily BTC-USD returns collected from September~7,
2021 to September~6, 2026 to fit Gaussian, Student-$t$, and symmetric
stable Paretian distributions. Mandelbrot \cite{mandelbrot1963variation} proposes stable Paretian distributions
for returns to account for departures from Gaussian behavior,
such as unusually large price movements.
Blattberg and Gonedes \cite{blattberg1974comparison} finds that empirical returns are more
Student-$t$ than symmetric stable Paretian, while
Chu et al.\ \cite{chu2015statistical} reports a better fit for Student-$t$ than Gaussian.
We select as the most suitable distribution the one with lowest AIC \cite{akaike1974new}---or lowest BIC \cite{schwarz1978estimating}. For completeness sake, we  compare the liquidation and repayment shortfall
probabilities obtained under the three fitted distributions.
In the Gaussian case, we also show that successive approximations of the liquidation and repayment shortfall probabilities of Theorems~\ref{thm:iid_increment_liquidation_probability} and~\ref{thm:iid_lender_shortfall_exact} approach the corresponding exact probabilities of Theorems~\ref{thm:borrower_liquidation_probability} and~\ref{thm:lender_repayment_shortfall_probability} as we increase truncation levels.

We estimate the parameters of the Gaussian and Student-$t$
distributions by maximum likelihood. For the symmetric stable Paretian
distribution, we perform maximum likelihood estimation using SciPy
\cite{2020SciPy-NMeth}, which implements Nolan's numerical method
\cite{nolan1997numerical}.
Table~\ref{tab:btc_return_mle} reports the maximum-likelihood estimates
$\widehat{\alpha}_1,\ldots,\widehat{\alpha}_8$ of the distribution
parameters $\alpha_1,\ldots,\alpha_8$, together with the corresponding estimated daily
drift $\mathbb E[\Delta Y_t]$, AIC, and BIC. Both AIC and BIC favor the Student-$t$ distribution.

In Figure~\ref{fig:3}, we consider an example loan with a duration of $d=30$ days, an annual effective interest rate of $6.00\%$, a required initial collateral ratio of $\rho=1.20$, and a liquidation threshold of $c=1.05$. The initial collateral is set at the minimum admissible level such that $\phi_0(z)=\rho v$, and the daily interest rate is $r=(1.06)^{1/365}-1$. Figure~\ref{fig:btc_liquidation_probability_comparison} compares the liquidation probabilities obtained from Theorem~\ref{thm:iid_increment_liquidation_probability} under the three return distributions from Table~\ref{tab:btc_return_mle}, with truncation levels $U=\Omega=1000$ for the inner and outer integrals in \eqref{eq:positive_partial_sum_cf_four_term} and \eqref{eq:exact_non_gaussian_liquidation_probability}. For short time horizons, both Student-$t$ and symmetric stable Paretian distributions give higher estimated liquidation probabilities than Gaussian distribution, which is consistent with the fact that they have heavier tails. The symmetric stable Paretian distribution gives the highest estimates for the liquidation probability throughout the thirty-day horizon, reaching a liquidation probability of $34.55\%$ within thirty days from loan origination. The Gaussian and Student-$t$ gives similar estimates for the liquidation probability, especially for long time horizons: the liquidation probabilities are $31.03\%$ and $30.89\%$, respectively, within 30 days. Next, we examine the numerical approximation in Theorem~\ref{thm:iid_increment_liquidation_probability} in the Gaussian case, for which Theorem~\ref{thm:borrower_liquidation_probability} provides an exact analytical baseline. We approximate \eqref{eq:exact_non_gaussian_liquidation_probability} at truncation levels $U=\Omega\in\{10,100,1000\}$. Figure~\ref{fig:btc_gaussian_cutoff_comparison} shows that, as the truncation levels increase, the numerical estimates converge to the analytical Gaussian probabilities as expected.

Figure~\ref{fig:btc_lender_shortfall_distribution_comparison} compares the repayment shortfall probabilities obtained from Theorem~\ref{thm:iid_lender_shortfall_exact} under the three return distributions from Table~\ref{tab:btc_return_mle}, with truncation levels $U=\Omega=1000$. The Gaussian, Student-$t$, and symmetric stable Paretian distributions give estimated repayment shortfall probabilities of $0.90\%$, $4.92\%$, and $9.90\%$, respectively, within 30 days. Using the same fitted Gaussian distribution and loan parameters, we next evaluate the repayment shortfall probability using Theorem~\ref{thm:iid_lender_shortfall_exact}. Figure~\ref{fig:btc_lender_shortfall_gaussian_cutoff_comparison} compares numerical evaluations of \eqref{eq:lender_exact_probability_im} at $U=\Omega\in\{100,1000\}$ with the analytical Gaussian probabilities of Theorem~\ref{thm:lender_repayment_shortfall_probability}.The case $U=\Omega=10$ is omitted because the numerical estimates deviate substantially from the analytical Gaussian probabilities. The integral defining $\psi_t$ in \eqref{eq:positive_partial_sum_cf_four_term} and the $u$ integrals in \eqref{eq:lender_exact_probability_im} are truncated at $U$, while the outer $\omega$ integral is truncated at $\Omega$. The repayment shortfall probability is lower than the liquidation probability throughout the thirty-day horizon. For example, on day~30, the Gaussian benchmarks give a repayment shortfall probability of $0.90\%$, compared with a liquidation probability of $31.04\%$.

\begin{table}[htbp]
\centering
\caption{Parameter estimates and model selection.}
\label{tab:btc_return_mle}
\small
\setlength{\tabcolsep}{2.5pt}
\renewcommand{\arraystretch}{1.25}

\sbox{\btcMLEtablebox}{%
    \begin{tabular}{@{}llrrr@{}}
        \toprule
        Distribution
        & \begin{tabular}[t]{@{}l@{}}
            Parameter\\[1pt]
            estimates
          \end{tabular}
        & \begin{tabular}[t]{@{}r@{}}
            Daily drift estimates\\[1pt]
            $\mathbb E[\Delta Y_t]\;(\times10^{-4})$
          \end{tabular}
        & AIC
        & BIC \\
        \midrule
        \addlinespace[3pt]

        Gaussian
        & $\begin{aligned}[t]
            \widehat{\alpha}_1 &= 2.92\times10^{-4}\\
            \widehat{\alpha}_2 &= 2.72\times10^{-2}
          \end{aligned}$
        & $-1.32$
        & $-7967.82$
        & $-7956.80$ \\
        \addlinespace[3pt]
        \midrule
        \addlinespace[3pt]

        Student-$t$
        & $\begin{aligned}[t]
            \widehat{\alpha}_3 &= 2.86\\
            \widehat{\alpha}_4 &= 1.95\times10^{-4}\\
            \widehat{\alpha}_5 &= 1.71\times10^{-2}
          \end{aligned}$
        & $-0.35$
        & $\mathbf{-8301.38}$
        & $\mathbf{-8284.85}$ \\
        \addlinespace[3pt]
        \midrule
        \addlinespace[3pt]

        \begin{tabular}[t]{@{}l@{}}
            Symmetric stable\\
            Paretian
        \end{tabular}
        & $\begin{aligned}[t]
            \widehat{\alpha}_6 &= 1.55\\
            \widehat{\alpha}_7 &= 1.40\times10^{-2}\\
            \widehat{\alpha}_8 &= 2.35\times10^{-4}
          \end{aligned}$
        & $-0.75$
        & $-8262.35$
        & $-8245.82$ \\
        \addlinespace[3pt]
        \bottomrule
    \end{tabular}%
}

\begin{minipage}{\wd\btcMLEtablebox}
    \noindent\usebox{\btcMLEtablebox}
    \par\smallskip
    {\footnotesize\raggedright
    The parameters are: Gaussian, $\alpha_1$ (mean) and $\alpha_2$
    (standard deviation); Student-$t$, $\alpha_3$ (degrees of freedom),
    $\alpha_4$ (location), and $\alpha_5$ (scale); symmetric stable
    Paretian, $\alpha_6$ (stability index), $\alpha_7$ (scale), and
    $\alpha_8$ (location).
    \par}
\end{minipage}
\end{table}

In practice, we also recommend updating the
estimated probabilities at each  checkpoint (with the last realization of $Y_t$) for improved prediction accuracy. Although checkpoints occur at a daily frequency in our example, the same approach holds when checkpoints occur at the same frequency as price updates by an oracle.

\begin{figure}[htbp]
    \centering
    \captionsetup{font=small,margin=5mm}
    \captionsetup[subfigure]{
        font=small,
        labelfont=normalfont,
        labelformat=parens,
        labelsep=none,
        justification=centering,
        singlelinecheck=true,
        skip=6pt,
        margin=0pt
    }
    \renewcommand{\thesubfigure}{\Alph{subfigure}}

    \begin{subfigure}[t]{0.49\textwidth}
        \centering
        \includegraphics[width=\linewidth]{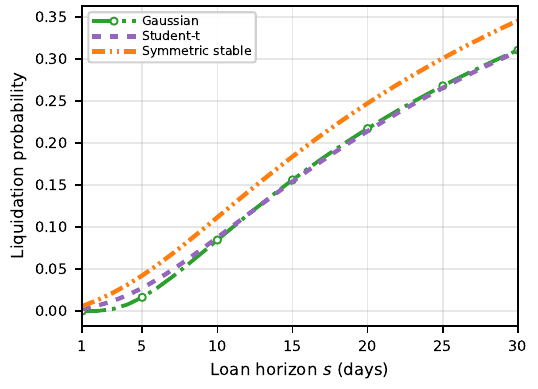}
        \caption{}
        \label{fig:btc_liquidation_probability_comparison}
    \end{subfigure}\hfill
    \begin{subfigure}[t]{0.49\textwidth}
        \centering
        \includegraphics[width=\linewidth]{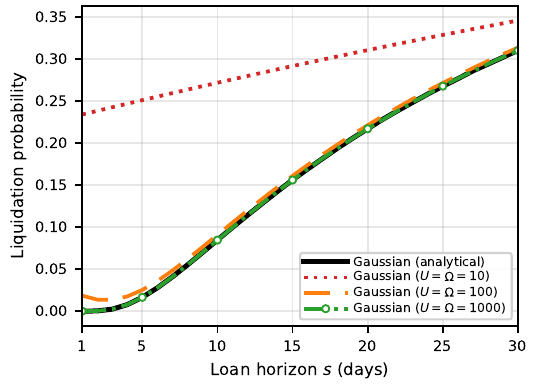}
        \caption{}
        \label{fig:btc_gaussian_cutoff_comparison}
    \end{subfigure}

    \par\medskip

    \begin{subfigure}[t]{0.49\textwidth}
        \centering
        \includegraphics[width=\linewidth]{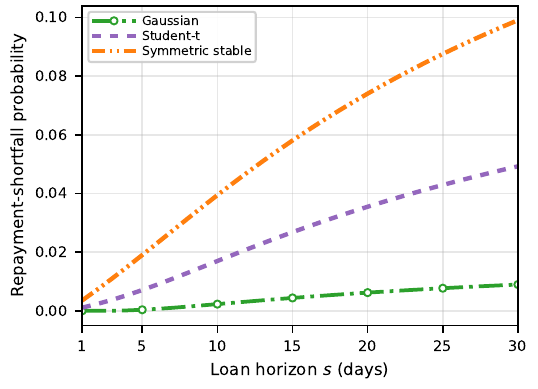}
        \caption{}
        \label{fig:btc_lender_shortfall_distribution_comparison}
    \end{subfigure}\hfill
    \begin{subfigure}[t]{0.49\textwidth}
        \centering
        \includegraphics[width=\linewidth]{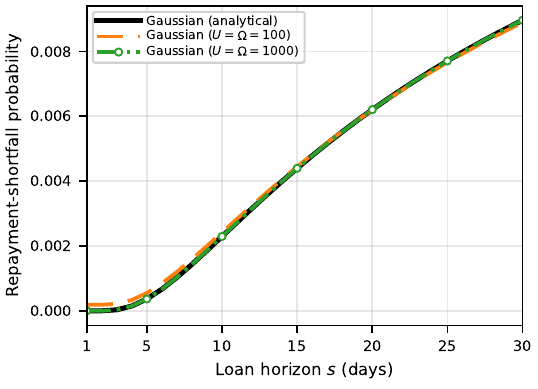}
        \caption{}
        \label{fig:btc_lender_shortfall_gaussian_cutoff_comparison}
    \end{subfigure}

    \caption{Liquidation probabilities (A, B) and repayment shortfall
    probabilities (C, D) over thirty days. Panels (A, C) compare Gaussian,
    Student-$t$, and symmetric stable Paretian return distributions
    using $U=\Omega=1000$. Panels (B, D) compare analytical Gaussian
    probabilities from Theorems~\ref{thm:borrower_liquidation_probability}
    and~\ref{thm:lender_repayment_shortfall_probability} with numerical
    estimates from Theorems~\ref{thm:iid_increment_liquidation_probability}
    and~\ref{thm:iid_lender_shortfall_exact}, respectively, using
    $U=\Omega\in\{10,100,1000\}$ in (B) and
    $U=\Omega\in\{100,1000\}$ in (D). All numerical estimates use a base
    quadrature step size of $0.5$.}
    \label{fig:3}
\end{figure}

\color{black}

\section{Related Works}\label{sec:related}

Our lending protocol has a unique set of features: it is a permissionless marketplace, with fixed rate and fixed maturity. Each of these features connects it to existing protocols, but together distinguish it from most.
By virtue of being a marketplace that matches lenders and borrowers to create fixed-rate fixed-term loans, it is reminiscent of traditional loans with well-defined fixed-rate, like the USD 28 trillion of US Treasury securities outstanding today \cite{finra}.
In the latter setting, there is a single borrower who issues bills through discriminatory and uniform-price auction \cite{BH15,goldreich} among lenders.
In our protocol, the interest rate of ongoing loans is determined by a different competition, where lenders and borrowers dynamically cancel and resubmit offers, and loans are taken-up instantaneously on a first-come-first-served basis, much like limit orders in order books \cite{cont2023mathematicalframeworkmodellingorder}.
Other examples of fixed-rate lending marketplaces include permissioned protocols, such as Prosper \cite{chen2014auctions} and LendingClub, and permissionless protocols, such as ETHLend \cite{ETHLend}. 

Our lending protocol is most similar to Morpho Midnight, which  introduces intent-based, fixed-term, fixed-rate markets in which lenders and borrowers express offers with acceptable terms \cite{morpho2025fixedrate,morpho2026midnight}.
Other fixed-rate and fixed-term lending protocols have recently been proposed, emphasizing under-collateralized borrowing \cite{wildcatdocs} or tokenized fixed-rate instruments \cite{termmaxdocs}. The main distinction of our protocol are the performance guarantees for both borrowers and lenders.



Other existing permissionless lending protocols are pool-based (e.g., Aave, Compound) with variable interest rates, or auction-based (Term Finance \cite{term}) with fixed rates.
In contrast to lending pools\footnote{For ease of comparison, we describe lending pools in Appendix~\ref{sec:pool} using the same section headings as Section~\ref{sec:marketplace}.}, our protocol is impervious to interest-rate manipulations by other users manipulating the utilization rate  \cite{hps,bl25} or by the lending-pool authority (e.g., Aave, Compound). Our interest rates are fixed, whereas pool-based lending protocols have variable interest rates that change with utilization \cite{bertucci2025agents} and a few lenders often control a majority of the liquidity \cite{gudgeon2020defiprotocolsloanablefunds}.
The interest rate for ongoing loans are products of a dynamic competition among participants, which is arguably more fair than the arbitrary rates \cite{hps} on lending pools---e.g., Chaos Labs set the interest rate for Aave until April 2026 \cite{chaos}.
A borrower can first look outside our protocol for the best available rate, and then publish an offer to borrow or immediately take up an offer to lend that is better than elsewhere.
Since our loans have a fixed duration $d$, we avoid redemption bottlenecks due to insufficient reserves in lending pool.
By setting a lower bound on allowed values for the duration $d$, our protocol can prevent flash loan, and the associated attacks \cite{qz}.

Most of the technical literature on lending protocols uses a continuous-time model (e.g., Brownian motion) for the valuation of tokens with respect to a num\'eraire (e.g., USD).
The probability of liquidation for a continuous-time lending protocol with zero drift is analyzed in \cite{belenko2025defiliquidationriskmodeling}.
The optimal collateral amount is derived in \cite{MP24}.
An alternative lending pool is proposed to address the impossibility of fair interest rates in existing lending pools \cite{hps}.
A dynamic equilibrium model is proposed in \cite{chiu} to predict equilibria that arise over time.
Discrete-time models similar to ours are presented in \cite{bastankhah2024thinkingfastslowdatadriven,bl25}.
The problem of slow convergence of interest rates to a market equilibrium is highlighted in \cite{bastankhah2024thinkingfastslowdatadriven}.
A detailed model of lending pools is presented in \cite{bl25} along with major features.

\section{Open Problems and Future Work}\label{sec:open}

This paper takes a first step in providing performance guarantees that enable rigorous risk management in permissionless decentralized finance.
There are of course many important directions for future work.

For instance, how do we generalize our results a setting with  stochastic gas fees?
How do we extend the performance guarantees of Theorem~\ref{thm:iid_increment_liquidation_probability} to returns that are not i.i.d.
Since the order of transactions on the blockchain are determined by gas fees, can we use a commitment scheme to protect participants from front-running \cite{kulkarni2023theorymaximalextractablevalue}.
How do we quantify the dynamics of participants that observe ongoing loan contracts and use t\^atonnement-like \cite{walras1874elements} offers to gauge market demand, as in \cite{cont2023mathematicalframeworkmodellingorder}.
How do we model the error between the valuation process $\phi_t$ and the actual proceeds of converting the collateral \textsc{col} to \textsc{liq}?

Although these guarantees are useful once two participants have already taken actions to contract a loan, other performance guarantees are beyond the scope of this paper.
We do not give guidance to the participants on when to take up an offer, versus creating an offer, or waiting for a better offer to appear. We do not characterize how quickly offers turn into loans, nor the convergence of interest rates of similar offers.

\bibliographystyle{alpha}
\bibliography{sample}

\clearpage
\appendix
\numberwithin{equation}{section}

\section{Figures and Tables}

\begin{table}[htbp]\caption{Notation for loan contract}
\begin{center}
\begin{tabular}{r c p{7cm} }
\toprule
\(w_L\) & \(\triangleq\) & lender wallet\\
\(w_B\) & \(\triangleq\) & borrower wallet\\
\(w_V\) & \(\triangleq\) & protocol-created vault wallet associated with this offer or loan\\
\(t_0\) & \(\triangleq\) & start of loan\\
\(v\) & \(\triangleq\) & \textsc{liq} principal amount\\
\(z\) & \(\triangleq\) & \textsc{col} collateral amount\\
\(d\) & \(\triangleq\) & duration of loan\\
\(r\) & \(\triangleq\) & interest rate for each time step\\
\(\rho\) & \(\triangleq\) & minimum required initial collateral-to-loan ratio at time \(t_0\)\\
\(c\) & \(\triangleq\) & liquidation threshold\\
\bottomrule
\end{tabular}
\end{center}
\label{tab:TableOfNotation}
\end{table}

\section{Proofs}
\label{app:proofs}
\subsection{Proof of Lemma~\ref{prop:multivariate_gaussian_representation}}
\label{app:proof_gaussian_representation}

\begin{proof}
By \eqref{eq:Y_gaussian_representation},
\[
    Y_i
    =y_0+i\log(1+r)-\sum_{k=1}^{i}\epsilon_k,
    \qquad i=1,\ldots,t.
\]
Thus, $(Y_1,\ldots,Y_t)^\top$ is an affine transformation of
the independent Gaussian vector
$(\epsilon_1,\ldots,\epsilon_t)^\top$ and is therefore jointly
Gaussian.

Its mean is
\begin{align*}
    \mathbb E[Y_i]
    &=y_0+i\log(1+r)-\sum_{k=1}^{i}\mu\\
    &=y_0+i\bigl(\log(1+r)-\mu\bigr).
\end{align*}
For $1\leq i,j\leq t$, independence of the increments gives
\begin{align*}
    \operatorname{Cov}(Y_i,Y_j)
    &=
    \sum_{k=1}^{i}\sum_{\ell=1}^{j}
    \operatorname{Cov}(\epsilon_k,\epsilon_\ell)\\
    &=
    \sum_{k=1}^{\min\{i,j\}}\sigma^2\\
    &=\sigma^2\min\{i,j\}.
\end{align*}
These are the stated mean vector and covariance matrix.
\end{proof}

\subsection{Proof of Theorem~\ref{thm:borrower_liquidation_probability}}
\label{app:proof_borrower_liquidation}

\begin{proof}
At the first checkpoint, liquidation occurs when $Y_1>0$. Therefore,
\begin{align}
    \mathbb P(T=1)
    =
    1-
    \Phi_1\left(
    0;
    \boldsymbol{\mu}_1^Y,
    \boldsymbol{\Sigma}_1^Y
    \right).
\end{align}
For each $t\in\{2,\ldots,d\}$, first liquidation at time $t$
requires that no liquidation occur through time $t-1$ and that
liquidation occur at time $t$. By
Lemma~\ref{prop:multivariate_gaussian_representation},
\begin{align}
    \begin{aligned}
    \mathbb P(T=t)
    &=
    \Phi_{t-1}\left(
    \mathbf 0_{t-1};
    \boldsymbol{\mu}_{t-1}^Y,
    \boldsymbol{\Sigma}_{t-1}^Y
    \right)-
    \Phi_t\left(
    \mathbf 0_t;
    \boldsymbol{\mu}_t^Y,
    \boldsymbol{\Sigma}_t^Y
    \right).
    \end{aligned}
\end{align}
The first CDF is the probability that
$Y_1,\ldots,Y_{t-1}$ are all nonpositive, so no liquidation occurs
through time $t-1$. The second CDF is the probability that
$Y_1,\ldots,Y_t$ are all nonpositive, so no liquidation occurs
through time $t$. Their difference is therefore the probability that
liquidation first occurs at time $t$.

Because the first-liquidation events are mutually disjoint, summing
their probabilities gives, for every $s\in\{2,\ldots,d\}$,
\begin{align}
    \begin{aligned}
    \mathbb P(T\leq s)
    &=
    1-
    \Phi_1\left(
    0;
    \boldsymbol{\mu}_1^Y,
    \boldsymbol{\Sigma}_1^Y
    \right)
    \\
    &\quad+
    \sum_{t=2}^{s}
    \Bigg[
    \Phi_{t-1}\left(
    \mathbf 0_{t-1};
    \boldsymbol{\mu}_{t-1}^Y,
    \boldsymbol{\Sigma}_{t-1}^Y
    \right)-
    \Phi_t\left(
    \mathbf 0_t;
    \boldsymbol{\mu}_t^Y,
    \boldsymbol{\Sigma}_t^Y
    \right)
    \Bigg]
    \\
    &=
    1-
    \Phi_s\left(
    \mathbf 0_s;
    \boldsymbol{\mu}_s^Y,
    \boldsymbol{\Sigma}_s^Y
    \right).
    \end{aligned}
\end{align}
The intermediate CDF terms cancel telescopically. Notice that the
result also applies to the case $s=1$. Consequently, we have
\eqref{eq:exact_borrower_liquidation_probability} holds for every
$s\in\{1,\ldots,d\}$.
\end{proof}
\subsection{Proof of Theorem~\ref{thm:lender_repayment_shortfall_probability}}
\label{app:proof_lender_repayment_shortfall}
\begin{proof}
Notice that no repayment shortfall can occur when $T=\infty$, which corresponds to a natural ending.
Section~\ref{sec:transition} specifies the terminal settlements under
both natural ending and liquidation. At time $d$, the protocol first
evaluates the liquidation condition using $\phi_d$. If the condition is
satisfied, the loan is liquidated; otherwise, the natural-ending
settlement is immediately executed using the same collateral valuation.
Therefore, $T=\infty$ implies that the liquidation condition is not
satisfied at time $d$, and \eqref{eq:decollateralization_time} gives
$Y_d\leq0$. Hence,
\begin{align}
    \phi_d(z)\geq cD_d>D_d.
\end{align}
Under the natural-ending settlement, the collateral transferred to the
lender consequently has value
\begin{align}
    \min\left\{D_d,\phi_d(z)\right\}=D_d
\end{align}
at time $d$. Thus, the lender is fully repaid at natural ending, and
the event $T=\infty$ cannot produce a repayment shortfall.

For a finite liquidation time $T=t$, the protocol instead applies the liquidation rule, the
lender receives only the amount of collateral needed to cover the
accrued debt, capped by the collateral available in the vault, and any
excess collateral is returned to the borrower. Under full-collateral
liquidation rule of Remark~\ref{rem:full}, all collateral held in the vault is transferred to the
lender, and no collateral is returned to the borrower. Accordingly, let $R_t^{\mathrm{ben}}$ and $R_t^{\mathrm{full}}$
denote the asset values, evaluated at liquidation time $t$,
of the collateral received by the lender under benevolent
and full-collateral liquidation, respectively. Then
\begin{align}
    R_t^{\mathrm{ben}}
    &=
    \min\left\{D_t,\phi_t(z)\right\},
    \quad
    R_t^{\mathrm{full}}
    =
    \phi_t(z).
\end{align}
Conditional on $T=t$, the shortfall condition under each rule satisfies
\begin{align}
    R_t^{\mathrm{ben}}<D_t
    &\Longleftrightarrow
    \min\left\{D_t,\phi_t(z)\right\}<D_t
    \Longleftrightarrow
    \phi_t(z)<D_t,
    \\
    R_t^{\mathrm{full}}<D_t
    &\Longleftrightarrow
    \phi_t(z)<D_t.
\end{align}
Hence, full-collateral liquidation rule generates
exactly the same repayment shortfall event. Because the two events
coincide, we may write $R_t$ for the asset value at time $t$ of the collateral received by the lender $R_t=R_t^{\mathrm{ben}}$ under the standard liquidation rule and $R_t=R_t^{\mathrm{full}}$ under full-collateral
liquidation rule. Their common repayment shortfall event at time $t$ is
\begin{align}
    \left\{T=t,\ \phi_t(z)<D_t\right\}
    =
    \left\{T=t,\ Y_t>\log c\right\}.
    \label{eq:common_repayment_shortfall_event}
\end{align}
The equality follows from
$Y_t=\log\left(cD_t/\phi_t(z)\right)$. On this common shortfall event,
$R_t=\phi_t(z)<D_t$ under either liquidation rule. If instead
$\phi_t(z)\geq D_t$, then $R_t\geq D_t$ under either rule, although the
realized proceeds may differ across the two rules. Thus, the choice of
liquidation rule may affect the amount recovered when no shortfall
occurs, but it does not affect the repayment shortfall event or its
probability conditional on $T=t$.

The liquidation boundary and the repayment shortfall boundary therefore
divide the state space into three regions:
\begin{align}
    \text{No liquidation; no repayment shortfall:}\quad
    &Y_t\leq0
    \quad\Longleftrightarrow\quad
    \phi_t(z)\geq cD_t>D_t,
    \\
    \text{Liquidation; no repayment shortfall:}\quad
    &0<Y_t\leq\log c
    \quad\Longleftrightarrow\quad
    D_t\leq\phi_t(z)<cD_t,
    \\
    \text{Liquidation with repayment shortfall:}\quad
    &Y_t>\log c
    \quad\Longleftrightarrow\quad
    \phi_t(z)<D_t.
\end{align}
The three-region partition summarizes the distinction between borrower
liquidation and lender repayment shortfall: the middle region permits
liquidation with full repayment. This completes the characterization of
repayment shortfall at every finite liquidation time $T=t$.

Using \eqref{eq:common_repayment_shortfall_event}, together with the
facts that $T=\infty$ does not produce a repayment shortfall and that
the finite first-liquidation events $\{T=t\}$, $t=1,\ldots,d$, are
mutually disjoint, we obtain, for every $s\in\{1,\ldots,d\}$,
\begin{align}
    &\mathbb P(\text{repayment shortfall by time }s)\\
    &=
    \sum_{t=1}^{s}
    \mathbb P\left(
        T=t,\ R_t<D_t
    \right)
    =
    \sum_{t=1}^{s}
    \mathbb P\left(
        T=t,\ Y_t>\log c
    \right).
    \label{eq:lender_shortfall_event_decomposition}
\end{align}

At the first checkpoint, crossing the repayment shortfall boundary
$Y_1>\log c>0$ also triggers the first liquidation. Therefore, the
$t=1$ term in \eqref{eq:lender_shortfall_event_decomposition} is
\begin{align}
    \mathbb P\left(
        T=1,\ Y_1>\log c
    \right)
    &=
    \mathbb P\left(
        Y_1>\log c
    \right)
    \nonumber\\
    &=
    1-
    \Phi_1\left(
        \log c;
        \boldsymbol{\mu}_1^Y,
        \boldsymbol{\Sigma}_1^Y
    \right).
\end{align}

For each $t\in\{2,\ldots,d\}$, the first-passage representation of $T$
gives
\begin{align}
    \mathbb P\left(
        T=t,\ Y_t>\log c
    \right)
    &=
    \mathbb P\left(
        Y_1\leq0,\ldots,Y_{t-1}\leq0,\,
        Y_t>\log c
    \right)
    \nonumber\\
    &=
    \Phi_{t-1}\left(
        \mathbf 0_{t-1};
        \boldsymbol{\mu}_{t-1}^Y,
        \boldsymbol{\Sigma}_{t-1}^Y
    \right)
    -
    \Phi_t\left(
        \left(
        \mathbf 0_{t-1}^{\top},
        \log c
        \right)^{\top};
        \boldsymbol{\mu}_t^Y,
        \boldsymbol{\Sigma}_t^Y
    \right).
\end{align}

Substituting these joint-event probabilities into
\eqref{eq:lender_shortfall_event_decomposition} gives the
conditional lender repayment shortfall probability by time $s$:
\begin{align}
    \begin{aligned}
    &\mathbb P(\text{repayment shortfall by time }s)\\
    &=
    1-
    \Phi_1\left(
        \log c;
        \boldsymbol{\mu}_1^Y,
        \boldsymbol{\Sigma}_1^Y
    \right)
    +
    \sum_{t=2}^{s}
    \Bigg[
    \Phi_{t-1}\left(
        \mathbf 0_{t-1};
        \boldsymbol{\mu}_{t-1}^Y,
        \boldsymbol{\Sigma}_{t-1}^Y
    \right)
    -
    \Phi_t\left(
        \left(
        \mathbf 0_{t-1}^{\top},
        \log c
        \right)^{\top};
        \boldsymbol{\mu}_t^Y,
        \boldsymbol{\Sigma}_t^Y
    \right)
    \Bigg].
    \end{aligned}
\end{align}
The right-hand side of
\eqref{eq:exact_lender_repayment_shortfall_probability_revised}
also agrees with the first-checkpoint result when $s=1$,
since the summation term is then absent. Consequently,
\eqref{eq:exact_lender_repayment_shortfall_probability_revised}
holds for every $s\in\{1,\ldots,d\}$.

\end{proof}

\subsection{Proof of Theorem~\ref{thm:iid_increment_liquidation_probability}}
\label{app:proof_iid_increment_liquidation_probability}

\begin{proof}
Fix $s\in\{1,\ldots,d\}$. Since $y_0<0$, the definitions of
$T$ and $M_s$ give
\begin{align}
    \{T\leq s\}
    =\bigcup_{t=1}^{s}\{Y_t>0\}
    =\left\{\max_{0\leq t\leq s}(Y_t-y_0)>-y_0\right\}
    =\{M_s>-y_0\}.
    \label{eq:non_gaussian_liquidation_maximum_event}
\end{align}
Thus, to compute the liquidation probability, it suffices to determine
the distribution of $M_s$ at $-y_0$.
We first determine the positive-part characteristic function $\psi_t$
for $t=1,\ldots,d$.
Recall that
$Y_t-y_0=\sum_{j=1}^{t}\Delta Y_j$, where the increments are i.i.d.
with characteristic function $f$. By independence, $f(\omega)^t$ is
the characteristic function of $Y_t-y_0$, so Pinelis's positive-part
formula \cite[Eq.~16]{pinelis2015characteristic} applies to
$(Y_t-y_0)_+$. This formula expresses $\psi_t(\omega)$ as
$\frac12[1+f(\omega)^t]$ plus the difference between two transform
values evaluated at $\omega$ and at zero. Using the positive-half-axis
representation following Eq.~(2) in \cite{pinelis2015characteristic},
we express each transform value as an integral and subtract the two
integrals to obtain
\begin{align*}
    \psi_t(\omega)
    =
    \frac{1+f(\omega)^t}{2}
    +
    \frac{1}{2\pi i}\int_0^\infty
    \frac{
        [f(\omega+u)^t-f(\omega-u)^t]
        -
        [f(u)^t-f(-u)^t]
    }{u}\,du.
\end{align*}
The integral above is well defined as an improper Riemann integral, since
Pinelis's Proposition~1 \cite[Proposition~1]{pinelis2015characteristic}
establishes the existence of its defining limit.
This establishes \eqref{eq:positive_partial_sum_cf_four_term}.

To derive \eqref{eq:spitzer_finite_recursion}, we pass from the
positive parts of the individual partial sums to their running maximum.
Since $Y_0-y_0=0$, we have
\begin{align*}
    M_s
    &=\max_{0\leq t\leq s}(Y_t-y_0)\\
    &=\max\left\{0,Y_1-y_0,\ldots,Y_s-y_0\right\}\\
    &=\max\left\{(Y_1-y_0)_+,\ldots,(Y_s-y_0)_+\right\}\\
    &=\max_{1\leq t\leq s}\left\{(Y_t-y_0)_+\right\},
    \qquad s=1,\ldots,d.
\end{align*}
Here, the positive part is taken separately for each partial sum
$Y_t-y_0$ before the maximum is taken.
In the present notation, the partial sums and their positive parts in
\cite{spitzer1956combinatorial} are $Y_t-y_0$ and $(Y_t-y_0)_+$,
respectively, and their running maximum is $M_s$.
Spitzer's Theorem~3.1 \cite[Eq.~(3.5)]{spitzer1956combinatorial}
gives the following generating-function identity for an infinite
i.i.d. sequence with the same increment law. Here we use only its
coefficients through degree $d$, which depend solely on
$\psi_1(\omega),\ldots,\psi_d(\omega)$:
\begin{align}
    \sum_{s=0}^{\infty}\chi_s(\omega)\zeta^s
    =
    \exp\!\left\{
        \sum_{t=1}^{\infty}
        \frac{\psi_t(\omega)}{t}\zeta^t
    \right\},
    \qquad |\zeta|<1.
    \label{eq:spitzer_generating_function}
\end{align}

Because $\chi_s$ and $\psi_t$ are characteristic functions, their
absolute values are bounded by one. The relevant power series therefore
converge absolutely and locally uniformly for $|\zeta|<1$, permitting
term-by-term differentiation and multiplication. Differentiating
\eqref{eq:spitzer_generating_function} with respect to $\zeta$ gives
\begin{align}
    \sum_{s=1}^{\infty}
    s\chi_s(\omega)\zeta^{s-1}
    =
    \exp\!\left\{
        \sum_{t=1}^{\infty}
        \frac{\psi_t(\omega)}{t}\zeta^t
    \right\}
    \sum_{t=1}^{\infty}
    \psi_t(\omega)\zeta^{t-1}
    =
    \left(
        \sum_{p=0}^{\infty}
        \chi_p(\omega)\zeta^p
    \right)
    \left(
        \sum_{t=1}^{\infty}
        \psi_t(\omega)\zeta^{t-1}
    \right).
    \label{eq:spitzer_differentiated_generating_function}
\end{align}

We now want the coefficient of $\zeta^{s-1}$ in the product of the
two series on the right-hand side of
\eqref{eq:spitzer_differentiated_generating_function}.
We find it using the Cauchy product: multiply one term from each
series and add their exponents of $\zeta$. Multiplying
$\chi_p(\omega)\zeta^p$ from the first series by
$\psi_t(\omega)\zeta^{t-1}$ from the second gives
$\chi_p(\omega)\psi_t(\omega)\zeta^{p+t-1}$.
For this term to contribute to the coefficient we want, its exponent
must satisfy $p+t-1=s-1$, or equivalently $p=s-t$.
Since the first series starts at $p=0$ and the second at $t=1$,
the possible values are $t=1,\ldots,s$.
Collecting the terms with the same power of $\zeta$, we obtain
\begin{align}
    &\left(
        \sum_{p=0}^{\infty}
        \chi_p(\omega)\zeta^p
    \right)
    \left(
        \sum_{t=1}^{\infty}
        \psi_t(\omega)\zeta^{t-1}
    \right)
    =
    \sum_{s=1}^{\infty}
    \left(
        \sum_{t=1}^{s}
        \psi_t(\omega)\chi_{s-t}(\omega)
    \right)
    \zeta^{s-1}.
    \label{eq:spitzer_cauchy_product}
\end{align}

We can now compare the coefficients on the two sides. After differentiation, the coefficient of $\zeta^{s-1}$ on the
left-hand side is $s\chi_s(\omega)$. On the right-hand side, it is
$\sum_{t=1}^{s}\psi_t(\omega)\chi_{s-t}(\omega)$. Since the two power series are equal for $|\zeta|<1$, uniqueness of
power-series coefficients gives
\begin{align}
    s\chi_s(\omega)
    =
    \sum_{t=1}^{s}
    \psi_t(\omega)\chi_{s-t}(\omega).
\end{align}
Since $M_0=0$, we have $\chi_0(\omega)=1$ and the discrete convolution recursion
\begin{align*}
    \chi_s(\omega)
    =\frac1s\sum_{t=1}^{s}
    \psi_t(\omega)\chi_{s-t}(\omega),
    \qquad s=1,\ldots,d.
\end{align*}
Thus, we have proved \eqref{eq:spitzer_finite_recursion}.

It remains to prove
\eqref{eq:exact_non_gaussian_liquidation_probability}. Let
$F_s(x)\triangleq\mathbb P(M_s\leq x)$ denote the distribution function
of $M_s$. By the theorem's assumption, $F_s$ is continuous at $-y_0$, and hence
$\mathbb P(M_s=-y_0)=0$. Moreover, since $M_s\geq0$ and $y_0<0$,
$F_s$ is also continuous at $y_0$, with $F_s(y_0)=0$.
Thus neither evaluation point carries probability mass, so we use
the continuity-point form of the Gil-Pelaez inversion formula \cite[p.~481]{gilpelaez1951inversion}. If there were mass at $-y_0$, the discontinuity version would instead
return the average of the two one-sided limits
\cite[p.~482]{gilpelaez1951inversion}. A corresponding formula for the
liquidation probability could then be derived by a similar argument,
but we do not consider that case here.

Under the present assumptions, the inversion formula gives, for
$x\in\{-y_0,y_0\}$,

\begin{align}
    F_s(x)
    =\frac12+\frac{1}{2\pi}\int_0^\infty
    \frac{
        e^{i\omega x}\chi_s(-\omega)
        -
        e^{-i\omega x}\chi_s(\omega)
    }{i\omega}\,d\omega.
    \label{eq:gil_pelaez_source_form}
\end{align}
Apply \eqref{eq:gil_pelaez_source_form} at $x=-y_0$ and $x=y_0$
using these continuity conditions, and subtract the latter identity
from the former. Since $M_s$ is real-valued, its characteristic function satisfies
$\chi_s(-\omega)=\overline{\chi_s(\omega)}$. Hence,
\[
\chi_s(-\omega)+\chi_s(\omega)
=2\operatorname{Re}\chi_s(\omega).
\]
Euler's formula gives
\[
e^{-iy_0\omega}-e^{iy_0\omega}
=-2i\sin(y_0\omega)
=2i\sin(-y_0\omega).
\]
Using these two identities, the difference between the two numerators becomes
\begin{equation}
\begin{aligned}
&\bigl[e^{-iy_0\omega}\chi_s(-\omega)
-e^{iy_0\omega}\chi_s(\omega)\bigr]
-\bigl[e^{iy_0\omega}\chi_s(-\omega)
-e^{-iy_0\omega}\chi_s(\omega)\bigr]\\
&\quad=
\bigl(e^{-iy_0\omega}-e^{iy_0\omega}\bigr)
\bigl(\chi_s(-\omega)+\chi_s(\omega)\bigr)\\
&\quad=
\bigl(2i\sin(-y_0\omega)\bigr)
\bigl(2\operatorname{Re}\chi_s(\omega)\bigr)\\
&\quad=
4i\sin(-y_0\omega)\operatorname{Re}\chi_s(\omega).
\end{aligned}
\label{eq:gil_pelaez_boundary_subtraction}
\end{equation}
It follows that
\begin{align}
    F_s(-y_0)-F_s(y_0)
    =
    \frac{2}{\pi}
    \int_{0}^{\infty}
    \operatorname{Re}\chi_s(\omega)
    \frac{\sin(-y_0\omega)}{\omega}\,d\omega.
\end{align}
Using $F_s(y_0)=0$, which follows from $M_s\geq0$ and $y_0<0$, we obtain
\begin{align}
    F_s(-y_0)
    =
    \frac{2}{\pi}
    \int_{0}^{\infty}
    \operatorname{Re}\chi_s(\omega)
    \frac{\sin(-y_0\omega)}{\omega}\,d\omega.
\end{align}
Finally, because
$\{T\leq s\}=\{M_s>-y_0\}$,
\begin{align}
    \mathbb P(T\leq s)
    =
    1-F_s(-y_0)
    =
    1-\frac{2}{\pi}
    \int_{0}^{\infty}
    \operatorname{Re}\chi_s(\omega)
    \frac{\sin(-y_0\omega)}{\omega}\,d\omega.
\end{align}
This establishes \eqref{eq:exact_non_gaussian_liquidation_probability}
and completes the proof.
\end{proof}

\subsection{Proof of Theorem~\ref{thm:iid_lender_shortfall_exact}}
\label{app:proof_iid_lender_shortfall_exact}

\begin{proof}
Fix $s\in\{1,\ldots,d\}$. Since $Y_t=\log(cD_t/\phi_t(z))$ and
$c>1$, a repayment shortfall at the first liquidation time $t$ occurs
exactly when the loan survives through $t-1$ and $Y_t>\log c$. Hence
\begin{equation*}
 \{T=t,\ \phi_t(z)<D_t\}
 =\{M_{t-1}\le-y_0,\ Y_t-y_0>\log c-y_0\}.
\end{equation*}
The probability of this event is the survival
probability multiplied by the conditional probability that the endpoint
exceeds the shortfall threshold. When $\mathbb P(M_{t-1}\le-y_0)>0$,

\begin{equation}
 \mathbb P(T=t,\ \phi_t(z)<D_t)
 =\mathbb P(M_{t-1}\le-y_0)\,
 \mathbb P(Y_t-y_0>\log c-y_0\mid M_{t-1}\le-y_0).
 \label{eq:lender_daily_probability_factorization}
\end{equation}
The first factor equals $1-\mathbb P(T\le t-1)$ and can therefore
be evaluated using \eqref{eq:exact_non_gaussian_liquidation_probability} in
Theorem~\ref{thm:iid_increment_liquidation_probability}, with value $1$
when $t=1$. To evaluate the second factor in
\eqref{eq:lender_daily_probability_factorization}, we compute the
characteristic function of the endpoint conditional on survival and
then apply Fourier inversion. This conditional
characteristic function is
\begin{equation}
 \mathbb E[e^{i\omega(Y_t-y_0)}\mid M_{t-1}\le-y_0]
 =\frac{\mathbb E[e^{i\omega(Y_t-y_0)}
 \mathbf1_{\{M_{t-1}\le-y_0\}}]}{\mathbb P(M_{t-1}\le-y_0)}.
 \label{eq:lender_conditional_cf}
\end{equation}
The denominator in \eqref{eq:lender_conditional_cf} is the known survival
probability. The numerator averages the endpoint exponential only over
paths satisfying the restriction on the past maximum; evaluating it
therefore requires the dependence between these two random quantities.

The proof follows three linked steps. Step~1 obtains the joint
characteristic function of $M_{t-1}$ and $Y_t-y_0$, which describes the
joint distribution needed in this restricted expectation. Step~2 uses
that function to impose $M_{t-1}\le-y_0$, evaluates the numerator in
\eqref{eq:lender_conditional_cf}, and divides by the known survival
probability to obtain the conditional characteristic function in
\eqref{eq:lender_conditional_cf}. Step~3 applies Fourier inversion to
\eqref{eq:lender_conditional_cf} at $\log c-y_0$ to evaluate the second
factor in \eqref{eq:lender_daily_probability_factorization}.
Multiplying by the survival probability and summing over $t=1,\ldots,s$
then gives the cumulative probability in
\eqref{eq:lender_exact_probability_im}.

\medskip\noindent\textbf{Step 1.}
The numerator in \eqref{eq:lender_conditional_cf} involves $M_{t-1}$
and $Y_t-y_0$. We first determine $\Gamma_{t-1}$, which describes
$(M_{t-1},Y_{t-1}-y_0)$, and then include the independent increment at
$t$ to obtain the joint characteristic function of the required pair.
Petrella and Kou \cite[Lemma~1, Eqs.~(7) and~(9)]{petrella2004numerical}
give the joint recurrence
\begin{align*}
 \Gamma_s(u,\omega)=\frac1s\sum_{t=1}^s
 \Big(
 \mathbb E[e^{i(u+\omega)(Y_t-y_0)_+}]
 +\mathbb E[e^{-i\omega\max\{y_0-Y_t,0\}}]-1
 \Big)\Gamma_{s-t}(u,\omega).
\end{align*}
The positive- and negative-part decomposition used in the proof of
their Lemma~2 yields the identity
\begin{equation*}
 e^{-i\omega\max\{y_0-Y_t,0\}}
 =1+e^{i\omega(Y_t-y_0)}-e^{i\omega(Y_t-y_0)_+}.
\end{equation*}
Taking expectations gives
$\mathbb E[e^{-i\omega\max\{y_0-Y_t,0\}}]
=1+f(\omega)^t-\psi_t(\omega)$, which establishes
\eqref{eq:lender_joint_recursion}.
Independence of $\Delta Y_t$ from the increments through $t-1$ now gives
\begin{align*}
 \mathbb E[e^{iuM_{t-1}+i\omega(Y_t-y_0)}]&=\mathbb E[e^{iuM_{t-1}+i\omega(Y_{t-1}-y_0)}]\mathbb E[e^{i\omega\Delta Y_t}]\\
 &=f(\omega)\Gamma_{t-1}(u,\omega).
\end{align*}
This is the joint transform of the two random quantities appearing in
\eqref{eq:lender_daily_probability_factorization}. Step~2 uses it to retain only the paths
that satisfy the survival restriction.

\medskip\noindent\textbf{Step 2.}
We now compute the expectation in the numerator of
\eqref{eq:lender_conditional_cf}. The same independence
used in Step~1 gives
\begin{align}
 &\mathbb E[e^{i\omega(Y_t-y_0)}\mathbf1_{\{M_{t-1}\le-y_0\}}]
 =f(\omega)\,
 \mathbb E[e^{i\omega(Y_{t-1}-y_0)}\mathbf1_{\{M_{t-1}\le-y_0\}}].
 \label{eq:lender_endpoint_increment}
\end{align}

We evaluate the remaining expectation by following the truncated-integral
argument of Gil-Pelaez \cite{gilpelaez1951inversion}. The identity obtained below can alternatively be derived from classical
Fourier--Stieltjes inversion applied to the distribution of the past
maximum weighted by the endpoint exponential
\cite[Vol.~II, Chapter~XVI, \S~4, Theorem~(4.5)]{zygmund1959trigonometric}. To express its
survival indicator as an integral, we first establish the algebraic
identity for the integrand. Euler's formulas give
\begin{align*}
 e^{iuM_{t-1}}+e^{-iuM_{t-1}}
 &=\bigl[\cos(uM_{t-1})+i\sin(uM_{t-1})\bigr]
  +\bigl[\cos(uM_{t-1})-i\sin(uM_{t-1})\bigr]\\
 &=2\cos(uM_{t-1}).
\end{align*}
The imaginary terms cancel. Multiplying this identity by
$\sin[(-y_0)u]$ therefore gives
\begin{equation*}
 \sin[(-y_0)u]\bigl[e^{iuM_{t-1}}+e^{-iuM_{t-1}}\bigr]
 =2\sin[(-y_0)u]\cos(uM_{t-1}).
\end{equation*}
Next, the sine addition and subtraction formulas give
\begin{align*}
 \sin[(-y_0+M_{t-1})u]
 &=\sin[(-y_0)u]\cos(uM_{t-1})
   +\cos[(-y_0)u]\sin(uM_{t-1}),\\
 \sin[(-y_0-M_{t-1})u]
 &=\sin[(-y_0)u]\cos(uM_{t-1})
   -\cos[(-y_0)u]\sin(uM_{t-1}).
\end{align*}
Adding these equations cancels the terms
$+\cos[(-y_0)u]\sin(uM_{t-1})$ and
$-\cos[(-y_0)u]\sin(uM_{t-1})$. Consequently,
\begin{equation*}
 \sin[(-y_0+M_{t-1})u]+\sin[(-y_0-M_{t-1})u]
 =2\sin[(-y_0)u]\cos(uM_{t-1}).
\end{equation*}
Combining the two identities, dividing by $\pi u$, and integrating
from $0$ to a finite $U>0$, we obtain
\begin{align*}
 &\frac1\pi\int_0^U\frac{\sin[(-y_0)u]}u
   \bigl[e^{iuM_{t-1}}+e^{-iuM_{t-1}}\bigr]\,du\\
 &\qquad=\frac1\pi\int_0^U
   \frac{\sin[(-y_0+M_{t-1})u]}u\,du
   +\frac1\pi\int_0^U
   \frac{\sin[(-y_0-M_{t-1})u]}u\,du.
\end{align*}

We now evaluate the limits of the two sine integrals.
The deterministic identity used in the proof of Gil-Pelaez
\cite[p.~481]{gilpelaez1951inversion} is
\[
    \int_0^\infty \frac{\sin v}{v}\,dv=\frac{\pi}{2}.
\]
We apply this identity one sample path at a time. Once we fix the path, $M_{t-1}$ is a fixed nonnegative real number.
We can then find the two limits by changing variables in each
integral, as follows.

For the first integral, $M_{t-1}\ge0$ and $y_0<0$ imply
$-y_0+M_{t-1}>0$. With
$v=(-y_0+M_{t-1})u$, we have
\[
    u=\frac{v}{-y_0+M_{t-1}},
    \qquad
    du=\frac{dv}{-y_0+M_{t-1}},
    \qquad
    \frac{du}{u}=\frac{dv}{v}.
\]
Consequently, for every $U>0$,
\[
    \frac1\pi\int_0^U
    \frac{\sin[(-y_0+M_{t-1})u]}{u}\,du
    =
    \frac1\pi\int_0^{(-y_0+M_{t-1})U}
    \frac{\sin v}{v}\,dv.
\]
Since $(-y_0+M_{t-1})U\to\infty$ as $U\to\infty$,
the cited identity gives
\[
    \lim_{U\to\infty}\frac1\pi\int_0^U
    \frac{\sin[(-y_0+M_{t-1})u]}{u}\,du
    =
    \frac1\pi\int_0^\infty\frac{\sin v}{v}\,dv
    =\frac12.
\]

For the second integral, the sign of $-y_0-M_{t-1}$
determines the result.

If $M_{t-1}<-y_0$, then $-y_0-M_{t-1}>0$.
Taking $v=(-y_0-M_{t-1})u$ again gives $du/u=dv/v$, and hence
\[
    \frac1\pi\int_0^U
    \frac{\sin[(-y_0-M_{t-1})u]}{u}\,du
    =
    \frac1\pi\int_0^{(-y_0-M_{t-1})U}
    \frac{\sin v}{v}\,dv.
\]
The upper endpoint tends to infinity, so
\[
    \lim_{U\to\infty}\frac1\pi\int_0^U
    \frac{\sin[(-y_0-M_{t-1})u]}{u}\,du
    =
    \frac1\pi\int_0^\infty\frac{\sin v}{v}\,dv
    =\frac12.
\]

If $M_{t-1}>-y_0$, then $M_{t-1}+y_0>0$.
The oddness of the sine function gives
\[
    \sin[(-y_0-M_{t-1})u]
    =
    \sin[-(M_{t-1}+y_0)u]
    =
    -\sin[(M_{t-1}+y_0)u].
\]
Therefore, taking $v=(M_{t-1}+y_0)u$, we obtain
\begin{align*}
    \frac1\pi\int_0^U
    \frac{\sin[(-y_0-M_{t-1})u]}{u}\,du
    &=
    -\frac1\pi\int_0^U
    \frac{\sin[(M_{t-1}+y_0)u]}{u}\,du \\
    &=
    -\frac1\pi\int_0^{(M_{t-1}+y_0)U}
    \frac{\sin v}{v}\,dv.
\end{align*}
Here $(M_{t-1}+y_0)U\to\infty$, and thus
\[
    \lim_{U\to\infty}\frac1\pi\int_0^U
    \frac{\sin[(-y_0-M_{t-1})u]}{u}\,du
    =
    -\frac1\pi\int_0^\infty\frac{\sin v}{v}\,dv
    =-\frac12.
\]

Combining these calculations, for every sample path with
$M_{t-1}\ne-y_0$,
\begin{align*}
    \lim_{U\to\infty}\frac1\pi\int_0^U
    \frac{\sin[(-y_0+M_{t-1})u]}{u}\,du
    &=\frac12,\\
    \lim_{U\to\infty}\frac1\pi\int_0^U
    \frac{\sin[(-y_0-M_{t-1})u]}{u}\,du
    &=
    \begin{cases}
        \frac12, & M_{t-1}<-y_0,\\
        -\frac12, & M_{t-1}>-y_0.
    \end{cases}
\end{align*}

\color{black}

Thus their sum tends to $1$ when $M_{t-1}<-y_0$ and to $0$ when
$M_{t-1}>-y_0$. The boundary event $M_{t-1}=-y_0$ has probability
zero by the continuity assumption; for $t=1$, it is excluded by
$M_0=0<-y_0$. Therefore,
\begin{equation}
\begin{aligned}
 \mathbf1_{\{M_{t-1}\le-y_0\}}
 &=\lim_{U\to\infty}\frac1\pi\int_0^U
 \frac{\sin[(-y_0)u]}u
 \bigl[e^{iuM_{t-1}}+e^{-iuM_{t-1}}\bigr]\,du\\
 &=\frac1\pi\int_0^\infty
 \frac{\sin[(-y_0)u]}u
 \bigl[e^{iuM_{t-1}}+e^{-iuM_{t-1}}\bigr]\,du
 \qquad\text{almost surely}.
\end{aligned}
\label{eq:lender_indicator_inversion}
\end{equation}

Multiplying by $e^{i\omega(Y_{t-1}-y_0)}$ and taking expectations gives
\begin{equation}
\begin{aligned}
 &\mathbb E[e^{i\omega(Y_{t-1}-y_0)}\mathbf1_{\{M_{t-1}\le-y_0\}}]\\
 &=\frac1\pi\mathbb E\!\left[\int_0^\infty\frac{\sin[(-y_0)u]}u e^{i\omega(Y_{t-1}-y_0)}\bigl(e^{iuM_{t-1}}+e^{-iuM_{t-1}}\bigr)\,du\right].
\end{aligned}
\label{eq:lender_indicator_substitution}
\end{equation}
We now evaluate the expectation on the right-hand side of
\eqref{eq:lender_endpoint_increment} by interchanging the expectation
and the improper integral.

For fixed $\omega\in\mathbb R$ and finite $U>0$, the integrand is
jointly measurable: after extending $\sin[(-y_0)u]/u$ continuously to
$-y_0$ at $u=0$, it is a continuous function of $u$ and the random vector
$(M_{t-1},Y_{t-1}-y_0)$. It is also absolutely integrable on the finite
interval, since
\begin{align*}
 &\frac1\pi\mathbb E\!\left[\int_0^U\left|e^{i\omega(Y_{t-1}-y_0)}\frac{\sin[(-y_0)u]}u\bigl(e^{iuM_{t-1}}+e^{-iuM_{t-1}}\bigr)\right|du\right]\\
 &\qquad\le\frac1\pi\int_0^U2(-y_0)\,du=\frac{2(-y_0)U}{\pi}<\infty.
\end{align*}
Here we used $|e^{i\omega(Y_{t-1}-y_0)}|=1$,
$|e^{iuM_{t-1}}+e^{-iuM_{t-1}}|\le2$, and
$|\sin[(-y_0)u]|/u\le-y_0$ for $u>0$.
Fubini's theorem therefore applies on every finite interval.

The deterministic sine-integral estimate in the proof of Pinelis's
Proposition~1 \cite{pinelis2015characteristic} implies
\begin{equation*}
 \left|\frac1\pi\int_0^U\frac{\sin[(-y_0\pm M_{t-1})u]}u\,du\right|\le1
 \qquad\text{for either sign and every }U>0.
\end{equation*}
The preceding trigonometric identity and the triangle inequality give
\begin{equation*}
 \left|\frac1\pi\int_0^U\frac{\sin[(-y_0)u]}u e^{i\omega(Y_{t-1}-y_0)}\bigl(e^{iuM_{t-1}}+e^{-iuM_{t-1}}\bigr)\,du\right|\le1+1=2.
\end{equation*}
These normalized truncated integrals converge almost surely by
\eqref{eq:lender_indicator_inversion}, and the bound is independent of
$U$. Bounded convergence therefore implies convergence of their
expectations. On each finite interval, Fubini's theorem identifies
these expectations with the corresponding integrals of the expected
integrand. Consequently, the improper integral of the expected
integrand exists and equals the expectation of the improper integral in
\eqref{eq:lender_indicator_substitution}. This justifies the first equality below; the remaining
equalities follow from linearity of expectation and
\eqref{eq:lender_joint_cf}:
\begin{align}
 &\frac1\pi\mathbb E\!\left[\int_0^\infty\frac{\sin[(-y_0)u]}u e^{i\omega(Y_{t-1}-y_0)}\bigl(e^{iuM_{t-1}}+e^{-iuM_{t-1}}\bigr)\,du\right]\nonumber\\
 &=\frac1\pi\int_0^\infty\frac{\sin[(-y_0)u]}u\mathbb E\!\left[e^{i\omega(Y_{t-1}-y_0)}\bigl(e^{iuM_{t-1}}+e^{-iuM_{t-1}}\bigr)\right]du\nonumber\\
 &=\frac1\pi\int_0^\infty\frac{\sin[(-y_0)u]}u\left[\mathbb E[e^{iuM_{t-1}+i\omega(Y_{t-1}-y_0)}]+\mathbb E[e^{-iuM_{t-1}+i\omega(Y_{t-1}-y_0)}]\right]du\nonumber\\
 &=\frac1\pi\int_0^\infty\frac{\sin[(-y_0)u]}u[\Gamma_{t-1}(u,\omega)+\Gamma_{t-1}(-u,\omega)]\,du.
 \label{eq:lender_improper_expectation_interchange}
\end{align}
The inner integral thus converges for every real $\omega$.

Combining \eqref{eq:lender_indicator_substitution} and
\eqref{eq:lender_improper_expectation_interchange} with
\eqref{eq:lender_endpoint_increment} gives the required numerator:
\begin{equation}
 \mathbb E[e^{i\omega(Y_t-y_0)}\mathbf1_{\{M_{t-1}\le-y_0\}}]=\frac{f(\omega)}\pi\int_0^\infty\frac{\sin[(-y_0)u]}u[\Gamma_{t-1}(u,\omega)+\Gamma_{t-1}(-u,\omega)]\,du.
 \label{eq:lender_survival_weighted_cf}
\end{equation}

Dividing by the known survival probability in
\eqref{eq:lender_conditional_cf} therefore gives, whenever that
probability is positive,
\begin{align*}
 &\mathbb E[e^{i\omega(Y_t-y_0)}\mid M_{t-1}\le-y_0]\\
 &\qquad=\frac{f(\omega)}{\pi\mathbb P(M_{t-1}\le-y_0)}\int_0^\infty\frac{\sin[(-y_0)u]}u[\Gamma_{t-1}(u,\omega)+\Gamma_{t-1}(-u,\omega)]\,du.
\end{align*}
This determines the conditional characteristic function in
\eqref{eq:lender_conditional_cf}. Step~3 uses it to evaluate the second
factor in \eqref{eq:lender_daily_probability_factorization}.

\medskip\noindent\textbf{Step 3.}
To evaluate the second factor in
\eqref{eq:lender_daily_probability_factorization}, we apply Fourier
inversion at $\log c-y_0$ to the conditional characteristic function in
\eqref{eq:lender_conditional_cf}, as obtained in Step~2.
If $\mathbb P(M_{t-1}\le-y_0)>0$, the
conditional endpoint distribution is continuous at $\log c-y_0$.
Indeed,
\begin{equation*}
 \{M_{t-1}\le-y_0,\ Y_t-y_0=\log c-y_0\}
 \subseteq\{M_t=\log c-y_0\},
\end{equation*}
and the latter event has probability zero by assumption. 

Applying the inversion formula of Gil-Pelaez
\cite[p.~481]{gilpelaez1951inversion}
to the conditional characteristic function in
\eqref{eq:lender_conditional_cf}, and taking the complementary
probability at $\log c-y_0$, gives
\begin{align*}
    &\mathbb P\bigl(
        Y_t-y_0>\log c-y_0
        \mid M_{t-1}\le-y_0
    \bigr) \\
    &\quad=
    \frac12+
    \frac1\pi\int_0^\infty
    \frac1\omega
    \operatorname{Im}\!\left\{
        e^{-i\omega(\log c-y_0)}
        \mathbb E\!\left[
            e^{i\omega(Y_t-y_0)}
            \mid M_{t-1}\le-y_0
        \right]
    \right\}\,d\omega.
\end{align*}

Substituting \eqref{eq:lender_conditional_cf} into this expression yields
\begin{align*}
    &\mathbb P\bigl(
        Y_t-y_0>\log c-y_0
        \mid M_{t-1}\le-y_0
    \bigr) \\
    &\quad=
    \frac12+
    \frac1\pi\int_0^\infty
    \frac1\omega
    \operatorname{Im}\!\left\{
        e^{-i\omega(\log c-y_0)}
        \frac{
            \mathbb E\!\left[
                e^{i\omega(Y_t-y_0)}
                \mathbf1_{\{M_{t-1}\le-y_0\}}
            \right]
        }{
            \mathbb P(M_{t-1}\le-y_0)
        }
    \right\}\,d\omega.
\end{align*}

Notice that the survival probability
$\mathbb P(M_{t-1}\le-y_0)$ does not change as we integrate over
$\omega$. Here it is a positive real number, so we can take its
reciprocal outside both the imaginary part and the integral:
\begin{align*}
    &\mathbb P\bigl(
        Y_t-y_0>\log c-y_0
        \mid M_{t-1}\le-y_0
    \bigr) \\
    &\quad=
    \frac12+
    \frac{1}{\pi\,\mathbb P(M_{t-1}\le-y_0)}
    \int_0^\infty
    \frac1\omega
    \operatorname{Im}\!\left\{
        e^{-i\omega(\log c-y_0)}
        \mathbb E\!\left[
            e^{i\omega(Y_t-y_0)}
            \mathbf1_{\{M_{t-1}\le-y_0\}}
        \right]
    \right\}\,d\omega.
\end{align*}
Multiplying by $\mathbb P(M_{t-1}\le-y_0)$, as required by
\eqref{eq:lender_daily_probability_factorization}, gives
\begin{align*}
    &\mathbb P(T=t,\ \phi_t(z)<D_t)
    =\frac12\,\mathbb P(M_{t-1}\le-y_0) \\
    &\quad+
    \frac{\mathbb P(M_{t-1}\le-y_0)}
         {\pi\,\mathbb P(M_{t-1}\le-y_0)}
    \int_0^\infty
    \frac1\omega
    \operatorname{Im}\!\left\{
        e^{-i\omega(\log c-y_0)}
        \mathbb E\!\left[
            e^{i\omega(Y_t-y_0)}
            \mathbf1_{\{M_{t-1}\le-y_0\}}
        \right]
    \right\}\,d\omega.
\end{align*}

Cancelling the common factor in the coefficient of the integral,
we obtain
\begin{align}
    &\mathbb P(T=t,\ \phi_t(z)<D_t)
    =\frac12\,\mathbb P(M_{t-1}\le-y_0)
    \nonumber\\
    &\quad+
    \frac1\pi\int_0^\infty
    \frac1\omega
    \operatorname{Im}\!\left\{
        e^{-i\omega(\log c-y_0)}
        \mathbb E\!\left[
            e^{i\omega(Y_t-y_0)}
            \mathbf1_{\{M_{t-1}\le-y_0\}}
        \right]
    \right\}\,d\omega.
    \label{eq:lender_daily_inversion}
\end{align}
\color{black}
If the survival probability is zero, the daily shortfall probability
and the weighted expectation are both zero, so the last expression
remains valid without conditioning on a zero-probability event.

Using \eqref{eq:exact_non_gaussian_liquidation_probability} in
Theorem~\ref{thm:iid_increment_liquidation_probability}, we obtain
\begin{equation*}
 \mathbb P(M_{t-1}\le-y_0)
 =\frac2\pi\int_0^\infty\operatorname{Re}\chi_{t-1}(u)
 \frac{\sin[(-y_0)u]}u\,du.
\end{equation*}
For $t=1$, the same identity holds with $\chi_0=1$ by the Dirichlet
integral. Substituting this identity and
\eqref{eq:lender_survival_weighted_cf} into
\eqref{eq:lender_daily_inversion}, and summing over the disjoint events
$\{T=t,\phi_t(z)<D_t\}$, proves
\eqref{eq:lender_exact_probability_im}.

\end{proof}
\section{Lending pools}\label{sec:pool}

In order to compare with our proposed protocol, we present a precise model of lending pools, such as Aave and Compound.

\subsection{Tokens, wallets, oracle}
First, we assume that there are two base tokens of interest: \textsc{col} and \textsc{liq}, along with the corresponding credit and debit tokens.
In addition to these two base tokens, we also introduce corresponding credit-tokens\footnote{We partially follow the notation from \cite{bl25}.} \(0^c\) and \(1^c\) and debit-tokens \(0^d\) and \(1^d\).
These tokens and their balances are controlled by the lending pool alone: wallets cannot send them to other wallets.

We let time steps be denoted \(t=1,2,\ldots\).
The set of participants are a set of wallets \(\mathbb W\). We abstract away the cryptographic notions behind wallets, and denote the content of a wallet \(w\) at time \(t\) by
\begin{align}
    w_t = (w_t[0], w_t[0^c], w_t[0^d], w_t[1], w_t[1^c], w_t[1^d]) \in \mathbb R_+^6.
\end{align}
Here, \(w_t[0]\) denotes the \textsc{col} balance, and 
\(w_t[1]\) denotes the \textsc{liq} balance;
\(w_t[0^c]\) and \(w_t[1^c]\) denote the credits in \textsc{col} and \textsc{liq}, and \(w_t[0^d]\) and \(w_t[1^d]\) denote the debits in \textsc{col} and \textsc{liq}.
Credits represent collateral that is deposited into a lending pool, whereas debits track amounts borrowed.

We also define the pool reserves at time \(t\) as the vector:
\begin{align}
    \Lambda_t = (\Lambda_t[0], \Lambda_t[0^c], \Lambda_t[0^d], \Lambda_t[1], \Lambda_t[1^c], \Lambda_t[1^d]) \in \mathbb R_+^6.
\end{align}

We assume that there exists a price oracle that 
takes as input a snapshot \(\Xi_t \in \mathcal X\) of the external environment at time \(t\) (e.g., information from various DEXs and CEXs) and outputs a vector of prices, one for each token:
\begin{align}
    \pi_t = \pi(\Xi_t) \in \mathbb R_+^6.
\end{align}
For a token \(x\), we write \(\frac{\pi_t[x^c]}{\pi_t[x]}\) for the slippage-free exchange rate from the corresponding credit token \(x^c\) to the token \(x\). Equivalently, \(\frac{\pi_t[x]}{\pi_t[x^c]}\) is the slippage-free exchange rate from the token \(x\) to the corresponding credit token \(x^c\). Similarly, \(\frac{\pi_t[x^d]}{\pi_t[x]}\) is the exchange rate from the debit token \(x^d\) to the token \(x\), and \(\frac{\pi_t[x]}{\pi_t[x^d]}\) is the exchange rate from the token \(x\) to the debit token \(x^d\).

\subsection{State}
The state of the lending pool at time $t$ is composed of the balances of all the wallets participating in the lending pool, the pool reserve $\Lambda_t$, and the price vector $\pi(\Xi_t)$. Assuming that there is a constant number \(n\) of wallets participating, then the state at time \(t\) is
\begin{align}
    S_t^{\mathrm{pool}} \triangleq (w_t^{(1)}, \ldots, w_t^{(n)}, \Lambda_t, \pi(\Xi_t)).
\end{align}

Following the distinction between wallet balances and pool reserves, we include $\Lambda_t$ explicitly in the state, since the pool reserve determines the available liquidity and cannot be replaced by the ordinary token balances held in user wallets.

Initially, all credit- and debit-balances are zero.

\subsection{Actions}

Next, we describe the actions available to the wallets at each time step.
The following actions are available to wallets at time \(t\).
\begin{enumerate}
    \item Deposit: a wallet \(w\) sends an amount \(v\) of \(x\)-token into the pool and receives \(\frac{\pi_t[x]}{\pi_t[x^c]}v\) corresponding \(x^c\)-credit-tokens. If the value \(v\) is allowed, this results in the following state variable updates:
    \begin{align}
        w_{t+1}[x] &= w_t[x] - v,\\
        w_{t+1}[x^c] &= w_t[x^c] + \frac{\pi_t[x]}{\pi_t[x^c]}v.
    \end{align}
    The corresponding pool-side reserve update is
    \begin{align}
        \Lambda_{t+1}[x] &= \Lambda_t[x] + v.
    \end{align}
    Observe that this deposit assumes that there is no slippage in converting a token \(x\) to its \(x^c\)-credit-token.

    \item Borrow: a wallet \(w\) receives \(v\) units of \(x\)-token from the pool and is assigned a corresponding amount of \(x^d\)-debit-token.
    If the value \(v\) is allowed (enough \(x\)-token in the pool and wallet has sufficient collateral), this results in the following state variable updates:
    \begin{align}
        w_{t+1}[x] &= w_t[x] + v,\\
        w_{t+1}[x^d] &= w_t[x^d] + \frac{\pi_t[x]}{\pi_t[x^d]}v.
    \end{align}
    The corresponding pool-side reserve update is
    \begin{align}
        \Lambda_{t+1}[x] &= \Lambda_t[x] - v.
    \end{align}

    \item Repay: a wallet \(w\) sends \(v\) units of \(x\)-token and reduces its borrowing balance by the corresponding amount of \(x^d\)-debit-token.
    If the value \(v\) is allowed, this results in the following state variable updates:
    \begin{align}
        w_{t+1}[x] &= w_t[x] - v,\\
        w_{t+1}[x^d] &= w_t[x^d] - \frac{\pi_t[x]}{\pi_t[x^d]}v.
    \end{align}
    The corresponding pool-side reserve update is
    \begin{align}
        \Lambda_{t+1}[x] &= \Lambda_t[x] + v.
    \end{align}

    \item Redeem: a wallet \(w\) sends \(v^c\) units of \(x^c\)-credit-token and receives back the corresponding amount of \(x\)-token.
    If the value \(v^c\) is allowed, this results in the following state variable updates:
    \begin{align}
        w_{t+1}[x] &= w_t[x] + \frac{\pi_t[x^c]}{\pi_t[x]}v^c,\\
        w_{t+1}[x^c] &= w_t[x^c] - v^c.
    \end{align}
    The corresponding pool-side reserve update is
    \begin{align}
        \Lambda_{t+1}[x] &= \Lambda_t[x] - \frac{\pi_t[x^c]}{\pi_t[x]}v^c.
    \end{align}

    \item Liquidate: if the health factor (cf. \cite{bl25}) of a borrower \(m\) falls below 1, then a wallet \(w\) can liquidate an amount \(v\) of \(m\)'s \(x\)-loan and receive a corresponding amount of \(y^c\)-credit-tokens from \(m\). This corresponding amount is usually scaled by a factor \(\eta>1\) to incentivize liquidation and dis-incentivize borrowers from not putting enough collateral into the pool. If the value \(v\) is allowed, this results in the following state variable updates:
    \begin{align}
        w_{t+1}[x] &= w_t[x] - v,\quad (\mbox{repay \(v\) on behalf}),\\
        w_{t+1}[y^c] &= w_t[y^c] + \eta \frac{\pi_t[x]}{\pi_t[y^c]}v, \quad (\mbox{liquidation reward}),\\
        m_{t+1}[x^d] &= m_t[x^d] - \frac{\pi_t[x]}{\pi_t[x^d]} v,\quad (\mbox{reduce debit balance}),\\
        m_{t+1}[y^c] &= m_t[y^c] - \eta \frac{\pi_t[x]}{\pi_t[y^c]}v, \quad (\mbox{loss of collateral}).
    \end{align}
    The corresponding pool-side reserve update is
    \begin{align}
        \Lambda_{t+1}[x] &= \Lambda_t[x] + v.
    \end{align}

    \item Sending and receiving: a wallet can send and receive tokens \(0\) and \(1\) from another wallet, but cannot send credit- and debit-tokens. The balances are updated correspondingly.
\end{enumerate}

\begin{remark}[Blockchain network fees]
So far, for simplicity, we have omitted network transaction fees (e.g., gas) that are typically associated with each action taken by a wallet. These fees can be modeled as additive random variables.    
\end{remark}

\subsection{State transitions}

We are now in a position to describe how the state evolves from one time to another, subject to the actions of wallets.
In addition to the effect of actions on updating state variables, the state transition also keeps track of interest accumulation as follows.
For each token $x$, the lending pool calculates an interest rate from the utilization of the pool reserve. Let
\begin{align}
    \mathrm{Debt}_t[x]
    &\triangleq
    \sum_{w \in \mathbb W}
    \frac{\pi_t[x^d]}{\pi_t[x]} w_t[x^d],
\end{align}
denote the total outstanding debt denominated in units of the base token $x$ at time $t$. The utilization rate of the pool for token $x$ is then
\begin{align}
    U_t[x]
    &\triangleq
    \frac{\mathrm{Debt}_t[x]}{\Lambda_t[x]+\mathrm{Debt}_t[x]}.
\end{align}
A central decision-maker (e.g., Chaos Lab for Aave \cite{chaos}) determines coefficients $\alpha$ and $\beta$, and sets the interest rate to
\begin{align}
    I_t[x]
    &\triangleq
    \alpha U_t[x]+\beta.
\end{align}
The lending pool then updates credit- and debit-balances for every wallet $w$:
\begin{align}
    w_{t+1}[x^c] &= (1+I_t[x]) \cdot w_t[x^c],\\
    w_{t+1}[x^d] &= (1+I_t[x]) \cdot w_t[x^d].
\end{align}

\end{document}